\documentclass[11pt]{article} 
\usepackage[letterpaper,hmargin=1in,vmargin=1in]{geometry}

\def\ShowAuthNotes{1}
\usepackage[utf8]{inputenc}

\usepackage{float,url,hyperref,epic,eepic}
\usepackage{xcolor, enumerate}
\usepackage{graphicx}
\usepackage{amsmath, amsfonts, amssymb, amsthm, amscd}
\usepackage{calc}
\usepackage{algorithmic, algorithm}
\usepackage{listings}

\newcommand{\Pd}{\mathsf{P}}
\newcommand{\cA}{\mathcal{A}}\newcommand{\cB}{\mathcal{B}}
\newcommand{\cD}{\mathcal{D}}

\newcommand{\cH}{\mathcal{H}}

\newcommand{\cL}{\mathcal{L}}
\newcommand{\cM}{\mathcal{M}}

\newcommand{\cS}{\mathcal{S}}
\newcommand{\cU}{\mathcal{U}}
\newcommand{\cW}{\mathcal{W}}
\newcommand{\cY}{\mathcal{Y}}

\newcommand{\accept}{\text{accept}}
\newcommand{\reject}{\text{reject}}
\newcommand{\teps}{\ensuremath{\tilde{\varepsilon}}}

\DeclareMathOperator{\Size}{\mathrm{Size}}

\newcommand{\gH}{\ensuremath{g_{\text{H}}^{C,\alpha}}}
\newcommand{\gHs}{\ensuremath{g_{\text{H}}}}
\newcommand{\pH}{\ensuremath{p_{\text{H}}}}
\newcommand{\gUY}{\ensuremath{g_{\text{UY}}^{C,V,\alpha}}}
\newcommand{\gUYs}{\ensuremath{g_{\text{UY}}}}
\newcommand{\gYL}{\ensuremath{g_{\text{YL}}^{C,V,\alpha}}}
\newcommand{\gYLs}{\ensuremath{g_{\text{YL}}}}
\newcommand{\gYs}{\ensuremath{g_{\text{Y}}}}
\newcommand{\pYL}{\ensuremath{p_{\text{YL}}}}
\newcommand{\gUH}{\ensuremath{g_{\text{UH}}^{C,\alpha}}}
\newcommand{\gUHs}{\ensuremath{g_{\text{UH}}}}
\newcommand{\pUH}{\ensuremath{p_{\text{UH}}}}

\newcommand{\NTIME}{\mathsf{NTIME}}
\newcommand{\EXP}{\mathsf{EXP}}

\newcommand{\PiVerifyHist}{\ensuremath{\Pi^{\text{VerifyHist}}}}
\newcommand{\PiLB}{\ensuremath{\Pi^{\text{LB}}}}

\makeatletter
\newtheorem*{rep@theorem}{\rep@title}
\newcommand{\newreptheorem}[2]{%
\newenvironment{rep#1}[1]{%
 \def\rep@title{#2 \ref{##1}}%
 \begin{rep@theorem}}%
 {\end{rep@theorem}}}
\makeatother

\newcounter{hours}
\newcounter{minutes}
\newcommand{\printtime}{ %
        \setcounter{hours}{\time/60} %
        \setcounter{minutes}{\time-\value{hours}*60} %
        \ifthenelse{\value{hours}<10}{0}{}\thehours:%
        \ifthenelse{\value{minutes}<10}{0}{}\theminutes%
        } 

\theoremstyle{plain}
        \newtheorem{theorem}{Theorem}[section]
        \newreptheorem{theorem}{Theorem}
        \newtheorem{lemma}[theorem]{Lemma}
        \newreptheorem{lemma}{Lemma}
        
        \newtheorem{claim}[theorem]{Claim}

\theoremstyle{definition}
        \newtheorem{definition}[theorem]{Definition}

\theoremstyle{remark}
        \newtheorem{remark}[theorem]{Remark}

\DeclareMathOperator*{\Exp}{E}

\DeclareMathOperator*{\Wd}{\text{W1}}
\DeclareMathOperator*{\Wdr}{\stackrel{\longrightarrow}{\Wd}}
\DeclareMathOperator*{\Wdl}{\stackrel{\longleftarrow}{\Wd}}

\newcommand{\inu}{\ensuremath{\leftarrow}}
\DeclareMathOperator{\poly}{poly}
\newcommand{\eps}{\ensuremath{\varepsilon}}

\newcommand{\co}{\mathsf{co}}
\newcommand{\NP}{\mathsf{NP}}
\newcommand{\distNP}{\mathsf{distNP}}
\newcommand{\NPpoly}{\mathsf{NP}\mathrm{/poly}}
\newcommand{\BPP}{\mathsf{BPP}}

\newcommand{\PSPACE}{\mathsf{PSPACE}}
\newcommand{\IP}{\mathsf{IP}}
\newcommand{\AMpoly}{\mathsf{AM}\mathrm{/poly}}
\newcommand{\AM}{\mathsf{AM}}

\DeclareMathOperator{\SAT}{SAT}

\definecolor{DSgray}{cmyk}{0,0,0,0.7}

\definecolor{DSred}{cmyk}{0,0.7,0,0.7}

\ifnum\ShowAuthNotes=1
\newcommand{\Authornote}[2]{\noindent{\small\textcolor{DSgray}{\sf{
\textcolor{red}{[#1: #2]\marginpar{\textcolor{red}{\fbox{\Large !}}}}}}}}
\newcommand{\Authormarginnote}[2]{\marginpar{\parbox{2.2cm}{\raggedright\tiny \textcolor{red}{#1: #2}}}}
\else
\newcommand{\Authornote}[2]{}
\newcommand{\Authormarginnote}[2]{}
\fi

\begin{document}
\title{A New View on Worst-Case to Average-Case Reductions for NP Problems}

\author{
Thomas Holenstein\thanks{ETH Zurich, Department of Computer Science, 8092 Zurich, Switzerland. E-mail: {\tt thomas.holenstein@inf.ethz.ch}} \and
Robin K\"unzler\thanks{ETH Zurich, Department of Computer Science, 8092 Zurich, Switzerland. E-mail: {\tt robink@inf.ethz.ch}}}
\maketitle
\begin{abstract}
We study the result by Bogdanov and Trevisan (FOCS,
2003), who show that under reasonable assumptions, 
there is no non-adaptive % worst-case to average-case
reduction that bases the average-case hardness of an
$\NP$-problem on the worst-case complexity of an $\NP$-complete
problem.
We replace the hiding and the heavy samples protocol in [BT03] by employing the histogram verification protocol of Haitner, Mahmoody and Xiao (CCC, 2010), which proves to be very useful in this context. Once the histogram is verified, our hiding protocol is directly public-coin, whereas the intuition behind the original protocol inherently relies on private coins. 
\end{abstract}

\newpage
\tableofcontents

\newpage

\section{Introduction}
\textit{One-way functions} are functions that are easy to compute on any instance, and hard to invert on average. Assuming their existence allows the construction of a wide variety of secure cryptographic schemes. Unfortunately, it seems we are far from proving that one-way functions indeed exist, as this would imply $\BPP\neq\NP$. Thus, the assumption that $\NP\nsubseteq\BPP$, which states that there exists a worst-case hard problem in $\NP$, is weaker. The following question is natural: 
\begin{quote}
	\textbf{Question~1:} Does $\NP \nsubseteq \BPP$ 
	  imply that one-way functions (or other cryptographic 
		primitives) exist?
\end{quote}
A positive answer to this question implies that the security of the aforementioned cryptographic schemes can be based solely on the worst-case assumption $\NP \nsubseteq \BPP$.

Given a one-way function $f$ and an image $y$, the problem of finding a preimage $x\in f^{-1}(y)$ is an $\NP$-problem: provided a candidate solution $x$, one can efficiently verify it by checking if $f(x)=y$. In this sense, a one-way function provides an $\NP$ problem that is hard to solve on average, and Question~1 asks whether it can be based on worst-case hardness. 
Thus, the question is closely related to the study of \textit{average-case complexity}, and in particular to the set $\distNP$ of distributional problems $(L, \cD)$, where $L \in \NP$, and $\cD$ is an ensemble of efficiently samplable distributions over problem instances. We say that a $\distNP$ problem $(L, \cD)$ is hard if there is no efficient algorithm that solves the problem (with high probability) on instances sampled from $\cD$. In this setting, analogously to Question~1, we ask:
\begin{quote}
	\textbf{Question~2:} Does $\NP \nsubseteq \BPP$ 
		imply that there exists a hard problem in $\distNP$?
\end{quote}

A natural approach to answer Question~2 affirmatively is to give a so-called \textit{worst-case to average-case reduction} from some $\NP$-complete $L$ to $(L',\cD) \in \distNP$: such a reduction $R^O$ is a polynomial time algorithm with black-box access to an oracle $O$ that solves $(L', \cD)$ on average, such that $\Pr_R[R^O(x)=L(x)]\geq 2/3$. 
We say a reduction is \textit{non-adaptive} if the algorithm $R$ fixes all its queries to $O$ in the beginning (see Section~\ref{sec:PrelimWCtACRed} for a formal definition). Bogdanov and Trevisan \cite{BT06} (building on work by Feigenbaum and Fortnow \cite{FF93}) show that it is unlikely that a non-adaptive worst-to-average-case reduction exists:
\begin{quote}
	\textbf{Main result of \cite{BT06} (informal):}
	If there exists a non-adaptive worst-case to average-case
	reduction from an $\NP$-complete problem to a problem
	in $\distNP$, then $\NP\subseteq \co\NPpoly$.
\end{quote}
The consequence $\NP \subseteq \co\NPpoly$ implies a collapse of the polynomial hierarchy to the third level \cite{Yap83}, which is believed to be unlikely. 

The work of Impagliazzo and Levin \cite{IL90} and Ben-David et al.~\cite{BCGL92} shows that an algorithm that solves a problem in $\distNP$ can be turned (via a non-adaptive reduction) into an algorithm that solves the search version of the same problem. Thus, as inverting a one-way function well on average corresponds to solving a $\distNP$ search problem well on average, the result of \cite{BT06} also implies that Question~1 cannot be answered positively by employing non-adaptive reductions, unless the polynomial hierarchy collapses.

\subsection{Contributions of this Paper}
The proof of the main result in \cite{BT06} proceeds as follows. Assuming that there exists a non-adaptive worst-case to average-case reduction $R$ from an $\NP$-complete language $L$ to $(L',\cD)\in \distNP$, it is shown that $L$ and its complement both have a constant-round interactive proof with advice (i.e.~$L$ and $\overline{L}$ are in $\AMpoly$ according to Definition~\ref{def:InteractiveProtocol}). As $\AMpoly = \NPpoly$, this gives $\co\NP \subseteq \NPpoly$. The final $\AMpoly$ protocol consists of three sub-protocols: the heavy samples protocol, the hiding protocol, and the simulation protocol. 
Using the protocol to verify the histogram of a probability distribution by Haitner et al.~\cite{HMX10}, we replace the heavy samples protocol and the hiding protocol. 
Our protocols have several advantages. The heavy samples protocol becomes quite simple, as one only needs to read a probability from the verified histogram. Furthermore, once the histogram is verified, our hiding protocol is directly public-coin, whereas the intuition behind the original hiding protocol crucially uses that the verifier can hide its randomness from the prover. Our protocol is based on a new and different intuition and achieves the same goal. Clearly, one can obtain a public-coin version of the original hiding protocol by applying the Goldwasser-Sipser transformation \cite{GS86}, but this might not provide a different intuition. 
Finally, our protocols show that the histogram verification protocol of \cite{HMX10} is a very useful primitive to approximate probabilities using $\AM$-protocols.

\subsection{Related Work}
Recall that our Question~2 above asked if average-case hardness can be based on the worst-case hardness of an $\NP$-complete problem. The question if cryptographic primitives can be based on $\NP$-hardness was stated as Question~1.

\vspace{0.3cm}
\noindent\textit{Average-case complexity.}
We use the definition of $\distNP$ and average-case hardness from \cite{BT06}. The hardness definition is essentially equivalent to Impagliazzo's notion of heuristic polynomial-time algorithms \cite{Imp95}. 
We refer to the surveys of Impagliazzo \cite{Imp95}, Goldreich \cite{Gol97}, and Bogdanov and Trevisan \cite{BT06avg} on average-case complexity. 

\vspace{0.3cm}
\noindent\textit{Negative results on Question~2.}
Feigenbaum and Fortnow \cite{FF93} study a special case of worst-case to average-case reductions, called random self-reductions. Such a reduction is non-adaptive, and reduces $L$ to itself, such that the queries are distributed uniformly at random (but not necessarily independently). They showed that the existence of a random self-reduction for an $\NP$-complete problem is unlikely, as it implies $\co\NP \subseteq \NPpoly$ and the polynomial hierarchy collapses to the third level. This result generalizes to the case of non-adaptive  reductions from $L\in \NP$ to $L' \in \distNP$ where the queries are distributed according to a distribution $\Pd$ that does not depend on the input $x$ to the reduction, but only on the length of $x$. 

The study of random self-reductions was motivated by their use to design interactive proof systems and (program-) checkers\footnote{
	Checkers allow to ensure the correctness of a given program on an input-by-input basis. Formally, a checker is an efficient algorithm $C$ that, given oracle access to a program $P$ which is supposed to decide a language $L$, has the following properties for any instance $x$. Correctness: If $P$ is always correct, then $C^P(x) = L(x)$ with high probability. Soundness: $C^P(x) \in \{L(x),\bot\}$ with high probability. 
}. 
Checkers are introduced by Blum and Blum and Kannan \cite{Blu88, BK95}. Rubinfeld \cite{Rub90} shows that problems that have a random self-reduction and are downward self-reducible (i.e.~they can be reduced to solving the same problem on smaller instances) have a program checker.
Random self-reductions can be used to prove the worst-case to average-case equivalence of certain $\PSPACE$-complete and $\EXP$-complete problems \cite{STV01}. A long-standing open question is whether $\SAT$ is checkable. In this context, Mahmoody and Xiao \cite{MX10} show that if one-way functions can be based on $\NP$-hardness via a randomized, possibly adaptive reduction, then $\SAT$ is checkable. 

In the context of program checking, Blum et al.~\cite{BLR93} introduce the notion of self-correctors. A self-corrector is simply a worst-case to average-case reduction from $L$ to $(L',\cD)$, where $L=L'$. Clearly, a random self-reduction is also a self-corrector.

As discussed earlier, based on \cite{FF93}, Bogdanov and Trevisan \cite{BT06} show that the average-case hardness of a problem in $\distNP$ cannot be based on the worst-case hardness of an $\NP$-complete problem via non-adaptive reductions (unless the polynomial hierarchy collapses). In particular, this implies that $\SAT$ does not have a non-adaptive self-corrector (unless the polynomial hierarchy collapses). It is an important open question if the same or a similar result can be proved for adaptive reductions. 

Watson \cite{Wat10} shows that there exists an oracle $O$ such that there is no worst-case to average-case reduction for $\NP$ relative to $O$. Impagliazzo \cite{Imp11} then gives the following more general result: any proof that gives a positive answer to Question~2 must use non-relativizing techniques. More precisely, it is shown that there exists an oracle $O$ such that $\NP^O\nsubseteq\BPP^O$, and there is no hard problem in $\distNP^O$. Note that this does not rule out the existence of a worst-case to average-case reduction, as such reductions do not necessarily relativize. In particular, the result of Bogdanov and Trevisan \cite{BT06} also applies to reductions that are non-adaptive and do \textit{not} relativize: there is no such reduction, unless the polynomial hierarchy collapses.

%bogdanov trevisan rules out reductions s.t.: is black-box, does not relativize, is non-adaptive. for example, such a r-s-r could exist. 
%The reduction gets acces to a avgSAT oracle, and decides SAT whp. 
%Now do I get a reduction that works if get access to a avgSAT^O oracle, and decides SAT^O whp? Its not clear how i would do that in general...
\vspace{0.3cm}
\noindent\textit{Negative results on Question~1.}
This question goes back to the work of Diffie and Hellman \cite{DH76}. Even and Yacobi \cite{EY80} give a cryptosystem that is $\NP$-hard to break. However, their notion of security requires that the adversary can break the system in the worst-case (i.e.~for every key). Their cryptosystem can in fact be broken on most keys, as shown by Lempel \cite{Lem79}. It is now understood that breaking a cryptosystem should be hard on average, which is, for example, reflected in the definition of one-way functions. 

Brassard \cite{Bra79} shows that public-key encryption cannot be based on $\NP$-hardness in the following sense: under certain assumptions on the scheme, if breaking the encryption can be reduced to deciding $L$, then $L \in \NP \cap \co\NP$. In particular, if $L$ is $\NP$-hard this implies that $\NP=\co\NP$. Goldreich and Goldwasser \cite{GG98} show the same result under relaxed assumptions. 

To give a positive answer to Question~$1$, one can aim for a reduction from an $\NP$-complete problem to inverting a one-way function well on average (see for example \cite{AGGM06} for a formal definition). As discussed earlier, the work of Impagliazzo and Levin \cite{IL90} and Ben-David et al.~\cite{BCGL92} allows to translate the results of \cite{FF93} and \cite{BT06} to this setting. That is, there is no non-adaptive reduction from an $\NP$-complete problem $L$ to inverting a one-way function, unless the polynomial hierarchy collapses to the third level. Akavia et al.~\cite{AGGM06} directly use the additional structure of the one-way function to prove that the same assumption allows the stronger conclusion $\co\NP \subseteq \AM$, which implies a collapse of the polynomial hierarchy to the second level. 

Haitner et al.~\cite{HMX10} show that if constant-round statistically hiding commitment can be based on an $\NP$-complete problem via $O(1)$-adaptive reductions (i.e.~the reduction makes a constant number of query rounds), then $\co\NP \subseteq \AM$, and the polynomial hierarchy collapses to the second level. In fact, they obtain the same conclusion for any cryptographic primitive that can be broken by a constant-depth collision finding oracle (such as variants of collision resistant hash functions and oblivious transfer). They also obtain non-trivial, but weaker consequences for $\poly(n)$-adaptive reductions. 

Bogdanov and Lee \cite{BL13} explore the plausibility of basing homomorphic encryption on $\NP$-hardness. They show that if there is a (randomized, adaptive) reduction from some $L$ to breaking a homomorphic bit encryption scheme (that supports the evaluation of any sufficiently ``sensitive'' collection of functions), then $L\in \AM \cap \co\AM$. In particular, if $L$ is $\NP$-complete this implies a collapse of the polynomial hierarchy to the second level. 

\vspace{0.3cm}
\noindent\textit{Positive results.}
We only know few problems in $\distNP$ that have worst-case to average-case reductions where the worst-case problem is believed to be hard. Most such problems are based on lattices, and the most important two are the short integer solution problem (SIS), and the learning with errors problem (LWE). 

The SIS problem goes back to the breakthrough work of Ajtai \cite{Ajt96}. 
He gives a reduction from an approximate worst-case version of the shortest vector problem to an average-case version of the same problem, and his results were subsequently improved \cite{Mic04, MR04}. Many cryptographic primitives, such as one-way functions, collision-resistant hash functions, identification schemes, and digital signatures have been based on the SIS problem, and we refer to \cite{BLPRS13} for an overview.  
He gives a reduction from an approximate worst-case version of the shortest vector problem to an average-case version of the same problem, and his results were subsequently improved \cite{Mic04, MR04}. Many cryptographic primitives, such as one-way functions, collision-resistant hash functions, identification schemes, and digital signatures have been based on the SIS problem, and we refer to \cite{BLPRS13} for an overview.  
He gives a reduction from an approximate worst-case version of the shortest vector problem to an average-case version of the same problem, and his results were subsequently improved \cite{Mic04, MR04}. Many cryptographic primitives, such as one-way functions, collision-resistant hash functions, identification schemes, and digital signatures have been based on the SIS problem, and we refer to \cite{BLPRS13} for an overview.  

Regev \cite{Reg05} gives a worst- to average-case reduction for the LWE problem in the quantum setting. That is, an algorithm for solving LWE implies the existence of a quantum algorithm to solve the lattice problem. The work of Peikert \cite{Pei09} and Lyubashevsky and Micciancio \cite{LM09} makes progress towards getting a reduction that yields a classical worst-case algorithm. The first classical hardness reduction for LWE (with polynomial modulus) is then given by Brakerski et al.~\cite{BLPRS13}. A large number of cryptographic schemes are based on LWE, and we refer to Regev's survey \cite{Reg10}, and to \cite{BLPRS13} for an overview. 

Unfortunately, for all lattice-based worst-case to average-case reductions, the worst-case problem one reduces to is contained in $\NP\cap\co\NP$, and thus unlikely to be $\NP$-hard.
We note that several of these reductions (such as the ones of \cite{Ajt96, Mic04, MR04}) are adaptive. 

Gutfreund et al.~\cite{GST07} make progress towards a positive answer to Question~2: they give a worst-case to average-case reduction for $\NP$, but sampling an input from the  distribution they give requires quasi-polynomial time. Furthermore, for any fixed $\BPP$ algorithm that tries to decide $\SAT$, they give a distribution that is hard for that specific algorithm. Note that this latter statement does not give a polynomial time samplable distribution that is hard for \textit{any} algorithm. Unlike in \cite{FF93, BT06}, where the reductions under consideration get black-box access to the average-case oracle, the reduction given by \cite{GST07} is not black-box, i.e.~it requires access to the code of an efficient average-case algorithm. Such reductions (even non-adaptive ones) are not ruled out by the results of \cite{FF93, BT06}. 
Gutfreund and Ta-Shma~\cite{GT07} show that even though the techniques of \cite{GST07} do not yield an average-case hard problem in $\distNP$, they bypass the negative results of \cite{BT06}. Furthermore, under a certain derandomization assumption for $\BPP$, they give a worst-case to average-case reduction from $\NP$ to an average-case hard problem in $\NTIME(n^{O(\log n)})$.

\section{Preliminaries}
\label{sec:SamplingPreliminaries}
\subsection{Notation}
We denote sets using calligraphic letters $\cA, \cB, \ldots$, and we write capital letters $A, B, \ldots$ to denote random variables. For a set $\cS$, we use $x\leftarrow \cS$ to denote that $x$ is chosen uniformly from $\cS$. We denote probability distributions
on bitstrings by $\Pd$, and write $x\leftarrow \Pd$ if $x$ is chosen from $\Pd$. Also, we let $\Pd(x) := \Pr_{y\leftarrow \Pd}[x=y]$. 
For $n\in \mathbb{N}$ we let $(n):=\{0, 1, \ldots, n\}$ and $[n]:=\{1,2,\ldots, n\}$.

\subsection{Concentration Bounds}
\label{sec:ConcBounds}
We use several concentration bounds and first state the well-known Chernoff bound.
\begin{lemma}[Chernoff bound]
	\label{lem:ChernoffBound}
	Let $X_1, \ldots, X_k$ be independent random variables where for all $i$ we have $X_i \in \{0,1\}$ and $\Pr[X_i=1]=p$ for some $p \in (0,1)$. Define $\widetilde{X}:=
	\frac{1}{k}\sum_{i\in [k]} X_i$. 	
	Then for any $\eps >0$ it holds that
	\begin{align*}
		&\Pr_{X_1, \ldots, X_k}\left[\widetilde{X} \geq p+\eps\right]
			<\exp{\left( 
					-\frac{\eps^2k}{2}
				\right)}, \quad
		\Pr_{X_1, \ldots, X_k}\left[\widetilde{X} \leq p-\eps\right]
			<\exp{\left( 
					-\frac{\eps^2k}{2}
				\right)}.
	\end{align*}
\end{lemma}

Hoeffding's bound \cite{Hoe63} states that for $k$ independent random variables $X_1, \ldots, X_k$ that take values in some appropriate range, with high probability their sum is close to its expectation. 
\begin{lemma}[Hoeffding's inequality]
	\label{lem:HoeffdingBound}
		Let $X_1, \ldots, X_k$ be independent random variables
	with $X_i \in [a,b]$, define $\widetilde{X}:=
	\frac{1}{k}\sum_{i\in [k]} X_i$ and let 
	$p=\Exp_{X_1, \ldots, X_k}[\widetilde{X}]$. 
	Then for any $\eps > 0$ we have
	\begin{align*}
		&\Pr_{X_1, \ldots, X_k}\left[
		\widetilde{X} \geq p
			+ \eps			
		\right]
		\leq \exp\left(-\frac{2\eps^2k}	
																{(b-a)^2}\right), \\
		&\Pr_{X_1, \ldots, X_k}\left[
		\widetilde{X} \leq
			p- \eps			
		\right]
		\leq \exp\left(-\frac{2\eps^2k}	
																{(b-a)^2}\right).
		\end{align*}
\end{lemma}

\subsection{Interactive Proofs}
\label{sec:InteractiveProofs}

In an interactive proof, an all-powerful prover tries to convince a computationally bounded verifier that her claim is true. The notion of an interactive protocol formalizes the interaction between the prover and the verifier, and is defined as follows. 

\begin{definition}[Interactive Protocol]
\label{def:InteractiveProtocol}
Let $n \in \mathbb{N}$,
$V:\{0,1\}^{\ast} \times \{0,1\}^{\ast} \rightarrow \{0,1\}^{\ast} \cup \{\accept, \reject\}$, 
$P:\{0,1\}^{\ast} \rightarrow \{0,1\}^{\ast}$, and $k, \ell, m: \mathbb{N}\rightarrow \mathbb{N}$. A 
$k$-round interactive
protocol $(V,P)$ with message length $m$ and $\ell$ random coins between $V$ and $P$ on input $x\in \{0,1\}^n$ is defined as follows:
\begin{enumerate}
	\item Uniformly choose random coins $r \in \{0,1\}^{\ell(n)}$ 
				for $V$.
	\item Let $k:=k(n)$ and repeat the following for $i=0,1,\ldots,k-1$:
				\begin{enumerate}[a)]	
					\item $m_i := V(x,i,r,a_0, \ldots, a_{i-1}), m_i \in \{0,1\}^{m(n)}$
					\item $a_i := P(x,i,m_0, \ldots, m_i), a_i \in \{0,1\}^{m(n)}$
				\end{enumerate}
				Finally, we have $V(x,k,r,m_0,a_0, \ldots, m_{k-1},
				a_{k-1}) \in \{\accept, \reject\}$.
\end{enumerate}
We denote by $(V(r),P)(x) \in \{\accept, \reject\}$ the output of $V$ on random coins $r$ after an interaction with $P$.
We say that $(x,r,m_0,a_0, \ldots, m_{j},a_{j})$ is consistent for $V$ if for all $i \in (k-1)$
we have $V(x,i,r,a_0, \ldots, a_{i-1})=m_i$. Finally, 
if $(x,r,m_0,a_0, \ldots, m_{k-1},	a_{k-1})$ is not consistent
for $V$, then $V(x,k,r,m_0,a_0, \ldots, m_{k-1}, 
a_{k-1}) = \reject$. 
\end{definition}

We now define the classes $\IP$ and $\AM$ of interactive proofs. The definition of $\IP$ was initially given by Goldwasser, Micali, and Rackoff~\cite{GMR85}, and the definition of $\AM$ goes back to Babai~\cite{Bab85}. 

\begin{definition}[$\IP$, $\AM$, and $\AMpoly$]	
	\label{def:InteractiveProofs}
	The set
	\begin{align*}
		\IP\left(
			\begin{array}{lll}
					\textsf{rounds}& =    & k(n)	 \\ 
					\textsf{time}  & =    & t(n) 	 \\
					\textsf{msg size}  & =    & m(n) 	 \\
					\textsf{coins} & =    & \ell(n)\\
					\textsf{compl} & \geq & c(n)   \\ 
					\textsf{sound} & \leq & s(n)   \\ 
			\end{array}
		\right)
	\end{align*}
	contains the languages $L$ that admit a $k$-round
	interactive protocol $(V,P)$ with message length $m$ and 
	$\ell$ random coins, and the following properties:
	\begin{itemize}
		\item[]
	\noindent \textbf{Efficiency:} 
	$V$ can be computed by an algorithm such that
	for any $x \in \{0,1\}^*$ and $P^*$ the total running time
	of $V$ in $(V,P^*)(x)$ is at most $t(|x|)$.
	
	\noindent \textbf{Completeness: }
	$$
		x \in L \implies \Pr_{r \leftarrow \{0,1\}^{\ell(|x|)}}
					\left[ (V(r),P)(x) = \accept
					\right] 
					\geq c(|x|).
	$$
	
	\noindent \textbf{Soundness: }
	For any $P^*$ we have
	$$
		x \notin L \implies \Pr_{r \leftarrow \{0,1\}^{\ell(|x|)}}
					\left[ (V(r),P^*)(x) = \accept
					\right] 
					\leq s(|x|).
	$$
	\end{itemize}
	The set $\AM$ is defined analogously, with the additional
	restriction that $(V,P)$ is \textit{public-coin}, i.e.~for
	all $i$, $m_i$ is an independent uniform random string. 
	We sometimes omit the \textsf{msg size} and
	\textsf{coins} parameters from the notation, in which case
	they are defined to be at most \textsf{time}. If we omit the \textsf{time} parameter, it is defined to be $\poly(n)$.  We then let
	\begin{align*}
		\IP:=
		\IP\left(
			\begin{array}{lll}
					\textsf{rounds}& =    & \poly(n)	 \\ 
					\textsf{compl} & \geq & 2/3   \\ 
					\textsf{sound} & \leq & 1/3   \\ 
			\end{array}
		\right),		\quad
		\AM:=
		\AM\left(
			\begin{array}{lll}
					\textsf{rounds}& =    & 1	 \\ 
					\textsf{compl} & \geq & 2/3   \\ 
					\textsf{sound} & \leq & 1/3   \\ 
			\end{array}
		\right)
	\end{align*}
	The set $\AMpoly$ is defined like $\AM$, but the verifier
	is additionally allowed to use $\poly(n)$ bits of 
	non-uniform advice.\footnote{
	For a definition of turing machines with advice we refer, for example, to Arora and Barak's book~\cite{AB09}, Chapter~6.3.
	}
\end{definition}

Instead of writing (for example) $L \in \AM(\textsf{rounds}=k,\textsf{time}=t,\textsf{compl}\geq c, \textsf{sound}\leq s)$, we sometimes say that $L$ has a $k$-round public-coin interactive proof with completeness $c$ and soundness $s$, where the verifier runs in time $t$. 

Babai and Moran~\cite{BM88} showed that in the definition of $\AM$ above, setting $\textsf{rounds}=k$ for any constant $k\geq 1$ yields the same class. The same is true for $\AMpoly$, and it thus suffices to give a $k$-round protocol with advice for some constant $k$ to place a language in $\AMpoly$. 

\paragraph{Assuming deterministic provers.} 
For proving the soundness condition of an interactive proof, without loss of generalty we may assume that the prover is determinsitic: we consider the deterministic prover that always sends the answer which maximizes the verifier's acceptance probability. No probabilistic prover can achieve better acceptance probability. 

\paragraph{Interactive proofs for promise problems.}
\textit{Promise problems} generalize the notion of languages, and are defined as follows: a promise problem $\Pi = (\Pi_Y, \Pi_N)$ is a pair of sets $\Pi_Y, \Pi_N \subseteq \{0,1\}^*$ such that $\Pi_Y \cap \Pi_N = \emptyset$. Given a problem $\Pi$, we are interested in algorithms (or protocols) that accept instances in $\Pi_Y$ and reject instances in $\Pi_N$. In particular, we don't care about the algorithm's behavior on instances that are not in $\Pi_Y \cup \Pi_N$. 

Promise versions of the classes $\IP, \AM$, and $\AMpoly$ are defined in the obvious way by restricting the completeness condition to $x\in \Pi_Y$ and the soundness condition to $x \in \Pi_N$. 

\subsection{Histograms and the First Wasserstein Distance}
\label{sec:HistAndWasserstein}
We give the definitions of histograms and Wasserstein distance as given in \cite{HMX10}. The histogram of a probability distribution $\Pd$ is a function $h:[0,1] \rightarrow [0,1]$ such that $h(p)=\Pr_{x\leftarrow \Pd}[\Pd(x)=p]$. The following definition describes a discretized version of this concept.
\begin{definition}[$(\eps, t)$-histogram]
	\label{def:histogram}
	Let $\Pd$ be a probability distribution on $\{0,1\}^n$, fix
	$t\in \mathbb{N}$, and let $\eps > 0$. 
	For $i \in (t)$ we define the $i$'th interval $\cA_i$ and the
	$i$'th bucket $\cB_i$ as
	\begin{align*}
			\cA_i := \left(2^{-(i+1)\eps}, 
								2^{-i\eps}\right], \qquad
			\cB_i := \left\{x: \Pd(x) \in \cA_i\right\}.
	\end{align*}
	We then let $h:=(h_0, \ldots, h_t)$ where 
	$h_i := \Pr_{x\leftarrow \Pd}[x \in \cB_i] = 
					\sum_{x\in \cB_i}\Pd(x)$. The tuple
	$h$ is called the $(\eps,t)$-histogram of $\Pd$.
\end{definition}
If for all $x$ we have either $\Pd(x)=0$ or 
$\Pd(x)\geq 2^{-n}$, and we consider the $(\eps,t)$-histogram of $\Pd$ for $t=\lceil n/\eps\rceil$, then $\bigcup_{i\in (t)}\cB_i = \{0,1\}^n$ and $\sum_{i\in (t)}h_i = 1$. If smaller probabilities occur (e.g.~$\Pd(x)=2^{-2n}$ for some $x$), this sum is smaller than $1$. 

The following observation follows directly from the above definition:
\begin{claim}
	\label{claim:sizeBi}
	For all $i\in (t)$ we have $h_i 2^{i\eps} 
					\leq |\cB_i| \leq h_i 2^{(i+1)\eps}$.
\end{claim}

\begin{proof}
	Recall that $h_i = \sum_{x\in \cB_i}\Pd(x)$ and
	by definition of $\cB_i$ we have
	$$ 
	|\cB_i|2^{-(i+1)\eps}
		=
			\sum_{x\in \cB_i}2^{-(i+1)\eps} 
		\leq
			\sum_{x\in \cB_i}\Pd(x) 
		\leq 
			\sum_{x\in \cB_i}2^{-i\eps} 
			= |\cB_i|2^{-i\eps}.\qedhere 
	$$
\end{proof}

We next introduce the Wasserstein distance between histograms. 
Intuitively, it measures how much work it takes to turn one histogram into another one, where the work is defined as the mass that is moved times the distance over which it is moved. We will only apply the Wasserstein distance to histograms $h$ where $\sum_{i\in (t)} h_i =1$, and we call such tuples \textit{distribution vectors}.  

\begin{definition}[1st Wasserstein distance over arrays]
	Given two distribution vectors
	$x$ and $y$ over $(t)$ we 
	let $a_i = \sum_{j\in (i)}x_j$ and $b_i = 
	\sum_{j\in (i)}y_j$. We let
	\begin{align*}
		\Wdr(x,y) := \frac{1}{t}\sum_{i\in (t): a_i > b_i}
								(a_i- b_i),\qquad 
		\Wdl(x,y) := \frac{1}{t}\sum_{i\in (t): b_i > a_i}
								(b_i- a_i),
	\end{align*}
	and $\Wd(x,y)  := \Wdr(x,y) + \Wdl(x,y)$.
	$\Wd(x,y)$ is called the 1st Wasserstein distance between 
	$x$ and $y$. $\Wdr(x,y)$ and $\Wdl(x,y)$ are called
	the right and left Wasserstein distance, respectively.
\end{definition}
In more general settings, this distance is also called Kantorovich distance, or Earth Mover's distance. For a more detailed discussion of this concept and the associated intuition we refer to the nice exposition in~\cite{HMX10}.

\subsection{The Parallel Lower Bound and Histogram Verification Protocols}
\label{sec:VerifyHistPLB}

To formalize the guarantees of the two protocols, we use the notion of promise problems as introduced in Section~\ref{sec:InteractiveProofs}.

\paragraph{The lower bound protocol.} 
We describe the promise problem that is solved by the parallel lower bound protocol as stated in \cite{BT06} (Corollary 7), which is based on the lower bound protocol of \cite{GS86}. The goal is to prove approximate lower bounds on the size of a set $\cS$ which is specified using a circuit $C$ as $\cS := C^{-1}(1) = \{x: C(x) = 1\}$. More generally, the following lemma states that there is a protocol that allows to prove lower bounds in parallel for several sets $\cS_i = C^{-1}(y_i)=\{x:C(x)=y_i\}$ for some given bit strings $y_i$. In the following, for a circuit $C$ we denote its size by $\Size(C)$. 

\begin{lemma}[Parallel Lower Bound Protocol, \cite{BT06}]
	\label{lem:PLBprotocol}
	For circuits $C:\{0,1\}^n \rightarrow \{0,1\}^m$,
	$\eps \in (0,1)$ we define the promise problem $\PiLB$ 
	as
		\begin{align*}
			\PiLB_Y &:= \left\{ 
					(C, \eps, y_1, s_1, \ldots, y_k, s_k) : 
					\forall i\in [k]: |C^{-1}(y_i)| \geq s_i
				\right\} \\
			\PiLB_N &:= \left\{
					(C, \eps, y_1, s_1, \ldots, y_k, s_k) :
					\exists i\in [k]: |C^{-1}(y_i)| \leq (1-\eps)s_i
			\right\}
		\end{align*}
	There exists a constant-round public-coin interactive proof
	for $\PiLB$ with completeness $1-\eps$ and soundness 
	$\eps$, where	the verifier runs in time 
	$\poly(\frac{\Size(C)k}{\eps})$. 
\end{lemma} 

We briefly sketch how such lower bounds can be proved, but refer to \cite{BT06} for a detailed exposition and a proof of the above lemma. Consider the case $k=1$, suppose the input $(C,\eps,s)$ is given, and we would like to give a protocol such that the verifier accepts with high probability if $|C^{-1}(1)|\geq s$ and rejects with high probability if $|C^{-1}(1)|\leq (1-\eps) s$. 
The protocol can be based on the hash mixing lemma:
\begin{lemma}[Hash Mixing Lemma, \cite{Nis92}]
		\label{lemma:hashMixingLemma}
		Let $\cB \subseteq \{0,1\}^n$, $x \in \{0,1\}^n$.
		If $\cH(n,m)$ is a family of $2$-wise independent hash functions mapping $n$ bits to $m$ bits, then the following holds. For all $\gamma > 0$ we have
				\begin{align*}
				   \Pr_{h\inu \cH(n,m)}
					 &\left[ |\{y\in \cB: h(y)=0^m\}|\notin
					 			 (1\pm\gamma)\frac{|\cB|}{2^m}
					 \right]\leq 
					 \begin{cases} 
					 		 \frac{2^m}{\gamma^2|\cB|}
							    &\mbox{if } |\cB|>0 \\
							0   &\mbox{if } |\cB|=0.
					 \end{cases}
				\end{align*}
\end{lemma} 
The idea is to let the verifier choose a pairwise independent hash function with an appropriate range, such that for the set $\cB:=C^{-1}(1)$ in case $|\cB| \geq s$, with high probability the set $\cM:=\{x: x\in \cB \land f(x)=0\}$ has size at least some fixed polynomial $p(n)$. One chooses the parameters such that in case $|\cB| \leq (1-\eps) s$, we have $|\cM| < p(n)$ with high probability. Then, the prover is supposed to send $p(n)$ many elements $x_1, \ldots, x_p$ to the verifier that satisfy $C(x_i)=1$ and $f(x_i)=0$. Finally, the verifier checks that the prover sent $p(n)$ many elements with these properties. It is not hard to see completeness and soundness, and Lemma~\ref{lem:PLBprotocol} can be proved for the parallel repetition of this protocol. 

\paragraph{Verifying histograms.}
We consider circuits $C: \{0,1\}^n \rightarrow \{0,1\}^m$, and the distribution $\Pd^C$ defined by $\Pd^C(y) = \Pr_{r\leftarrow \{0,1\}^n}[C(r) =y]$. The VerifyHist protocol of \cite{HMX10} (Lemma 4.4) allows to verify that some given histogram $h$ is close to the histogram of $\Pd^C$ in terms of the Wasserstein distance.
\begin{lemma}[VerifyHist Protocol, \cite{HMX10}]	
	\label{lem:VerifyHist}
	For a circuit $C:\{0,1\}^n \rightarrow \{0,1\}^m$,
	$\eps \in (0,1)$, and $t = \lceil n/\eps \rceil$,	
	we denote by $h^C \in [0,1]^{t+1}$ the 
	$(\eps,t)$-histogram of $\Pd^C$. 
	We define the promise problem $\PiVerifyHist$ as
		\begin{align*}
			\PiVerifyHist_Y &:= \left\{ 
					(C,\varepsilon,h) : 
					h=h^C
				\right\} \\
			\PiVerifyHist_N &:= \left\{
					(C,\varepsilon,h) :
					\Wd(h^C,h) > 20/t
			\right\}
		\end{align*}
	There exists a constant-round public-coin interactive proof
	for $\PiVerifyHist$ with completeness $1-2^{-n}$ and
	soundness $2^{-n}$,  where the verifier runs in time 
	$\poly(\frac{\Size(C)}{\eps})$. 
\end{lemma}
We remark that $\Wd(h^C,h)$ is well-defined: for all $y$ we have $\Pd^C(y) =0$ or $\Pd^C(y)\geq 2^{-n}$, and thus by choice of $t$, $h^C$ is a distribution vector. 

We give a brief and intuitive description of the VerifyHist protocol, but refer to \cite{HMX10} for a formal treatment and a proof of the above lemma. 
For a circuit $C$ and a claimed histogram $h$ the protocol proceeds as follows. 

The first part of the protocol is called \textit{preimage test}: the verifier samples elements $y_1, \ldots, y_k$ (for some appropriate $k$) from the distribution $\Pd^C$ and sends them to the prover. The honest prover sends back the probabilities $\Pd^C(y_i)$, and proves a lower bound on them using the parallel lower bound protocol of Lemma~\ref{lem:PLBprotocol}. Finally, the verifier considers the histogram $h'$ induced by the values $\Pd^C(y_i)$ and accepts if and only if $\Wd(h,h')$ is small. 

In the second part, a so-called \textit{image test} is performed: let $\cW_i := \{y: \Pd^C(y) \geq 2^{-i\eps}\}$, and let $w^h_i$ be the estimates of $\cW_i$ given by the claimed histogram $h$. Using the parallel lower bound protocol, the prover proves that indeed $|\cW_i| \geq w^h_i$ for all $i$. 

Intuitively, the preimage test prevents the prover from claiming that many probabilities are larger than they actually are, and it can be shown that the image test rejects in case the probabilities are larger than claimed. Haitner et al.~\cite{HMX10} prove that indeed, if both tests accept, then $h$ is close to $h^C$ in the first Wasserstein distance. 

\subsection{Worst-Case to Average-Case Reductions}
\label{sec:PrelimWCtACRed}
We give a definition of non-adaptive worst-case to average-case reductions. Informally, such a reduction is a polynomial time algorithm that, given an oracle which solves some given problem on average, solves some other problem in the worst case. The reduction is called \textit{non-adaptive} if it generates all its oracle queries before calling the oracle. The definition we give is from \cite{BT06}.

%A set of distributions $\Pd = \{\Pd_n\}_{n\in \mathbb{N}}$ where $\Pd_n \in \{0,1\}^n$ is called \textit{efficiently samplable} if there exists a probabilitstic polynomial time algorithm that on input $1^n$ generates a sample of $\Pd_n$. The set $\distNP$ contains all pairs $(L,\Pd)$ where $L\in \NP$ and $\Pd$ is efficiently samplable.

A distributional problem is a pair $(L, \cD)$, where $L$ is a language and $\cD$ is a set $\cD = \{\Pd_n\}_{n\in \mathbb{N}}$, and for each $n$, $\Pd_n$ is a distribution over $\{0,1\}^n$. 

%We say that a problem $(L, \Pd)$ is \textit{easy on average} if for every polynomial $p(n)$ there exists a probabilistic polynomial time algorithm $A$ such that $$\Pr_{x\leftarrow \Pd_n, A}[A(x)= L(x)] \geq 1-1/p(n).$$ 

\begin{definition}
	A non-adaptive $\delta$-worst-to-average reduction from $L$ 
	to a distributional problem $(L',\cD)$ is a family of 
	polynomial size	circuits $\{R_n\}_{n\in \mathbb{N}}$ such
	that for any $n$ the following holds:
	\begin{itemize}
		\item $R_n$ takes as input some $x\in\{0,1\}^n$ 
					and randomness $r$, and outputs 
					$(y_1, \ldots, y_k)$ (called queries), and a
					circuit $C$. 
		\item For any $x\in \{0,1\}^n$, and any oracle $O$ for
					which 
						$\Pr_{x\leftarrow \Pd_n, O}[O(x)
							\neq L'(x)]\leq \delta(n)$
					we have 
						$$ \Pr_{r, (y_1, \ldots, y_k, C):=R_n(x,r)}
						[C(O(y_1), \ldots, O(y_k))=L(x)]\geq 2/3.$$
	\end{itemize}
\end{definition}
We may assume that the queries $y_1, \ldots, y_k$ are identically (but not necessarily independently) distributed. If this is not the case for the original reduction $R$, we can easily obtain a reduction $R'$ that satisfies this property: $R'$ obtains the queries of $R$ and outputs a random permutation of them (the circuit $C$ is also modified accordingly).

Furthermore, the constant $2/3$ can be replaced by $1/2 + 1/n^c$ for some constant $c$: by the usual repetition argument, executing the reduction a polynomial number of times and outputting the majority answer still yields an exponentially small error probability.

\section{Technical Overview}
\label{sec:TechnicalOverviewWTA}

For a formal definition of non-adaptive worst-case to average-case reductions, we refer to Section~\ref{sec:PrelimWCtACRed} in the preliminaries. In the introduction we stated an informal version of the result of \cite{BT06}. We now state their main theorem formally. Let $\cU$ be the set $\{\Pd_n\}_{n\in \mathbb{N}}$ where $\Pd_n$ is the uniform distribution on $\{0,1\}^n$. 

\begin{theorem}[Main Theorem of \cite{BT06}]
	For any $L$ and $L'$ and every constant $c$ the following
	holds.
	If $L$ is $\NP$-hard, $L' \in \NP$, and there exists a 
	non-adaptive $1/n^c$-worst-to-average reduction from 
	$L$ to $(L',\cU)$, then $\co\NP \subseteq \NPpoly$. 
\end{theorem}

As discussed earlier, the conclusion implies a collapse of the polynomial hierarchy to the third level. 

The theorem is stated for the set of uniform distributions $\cU$. Using the results of Ben-David et al.~\cite{BCGL92} and Impagliazzo and Levin~\cite{IL90}, the theorem can be shown to hold for any polynomial time samplable set of distributions $\cD$. This is nicely explained in \cite{BT06} (Section~5). 

We first give an overview of the original proof, and then describe how our new protocols fit in. 

\subsection{The proof of Bogdanov and Trevisan}
Suppose $R$ reduces the $\NP$-complete language $L$ to $(L', \cU)\in \distNP$. 
The goal is to give a (constant-round) $\AMpoly$ protocol for $L$ and its complement. As $\NPpoly=\AMpoly$, this will give the result. The idea is to simulate an execution of the reduction $R$ on input $x$ with the help of the prover. The verifier then uses the output of $R$ as its guess for $L(x)$. $R$ takes as input the instance $x$, randomness $r \in \{0,1\}^n$, and produces (non-adaptively) queries $y_1, \ldots, y_k \in \{0,1\}^m$ for the average-case oracle. The reduction is guaranteed to guess $L(x)$ correctly with high probability, provided the oracle answers are correct with high probability. 
As mentioned in Section~\ref{sec:PrelimWCtACRed}, we may assume that the queries $y_1, \ldots, y_k$ are identically (but not necessarily independently) distributed. We denote the resulting distribution of individual queries by $\Pd^{R,x}$, i.e.~$\Pd^{R,x}(y) = \Pr_{r}[R(x,r) = y]$ (where $R(x,r)$ simply outputs the first query of the reduction on randomness $r$). 

\paragraph{Handling uniform queries: the Feigenbaum-Fortnow protocol.}
The proof of \cite{BT06} relies on the following protocol by Feigenbaum and Fortnow \cite{FF93}. The protocol assumes that the queries are uniformly distributed, i.e.~$\Pd^{R,x}(y)=2^{-m}$ for all $y$. The advice for the $\AMpoly$ protocol is $\gUYs = \Pr_{y \inu \{0,1\}^m}[y \in L']$, i.e.~the probability of a uniform sample being a yes-instance. The protocol proceeds as follows. First, the verifier chooses random strings $r_1, \ldots, r_{\ell}$ and sends them to the prover. The honest prover defines $(y_{i1}, \ldots, y_{ik}):=R(x,r_i)$ for all $i$, and indicates to the verifier which $y_{ij}$ are in $L'$ (we call them yes-instances), and provides the corresponding $\NP$-witnesses. The verifier checks the witnesses, expects to see approximately a $\gUYs$ fraction of yes-answers, and rejects if this is not the case. The verifier then chooses a random $i$ and outputs $R(x,r_i, y_{i1}, \ldots, y_{ik})$ as its guess for $L(x)$. 

To see completeness, one uses a concentration bound to show that the fraction of yes-answers sent by the prover is approximately correct with high probability (one must be careful at this point, because the outputs of the reduction for a fixed $r_i$ are not independent). Finally, the reduction decides $L(x)$ correctly with high probability. 

To argue that the protocol is sound, we note that the prover cannot increase the number of yes-answers at all, as it must provide correct witnesses. Furthermore, the prover cannot decrease the number of yes-answers too much, as the verifier wants to see approximately a $\gUYs$ fraction. This gives that most answers provided by the prover are correct, and thus with high probability the reduction gets good oracle answers, in which case it outputs $0$ with high probability. 

We note that the Feigenbaum-Fortnow simulation protocol is public-coin.

\paragraph{The case of smooth distributions: the Hiding Protocol.}	
Bogdanov and Trevisan \cite{BT06} generalize the above protocol so that it works for distributions that are $\alpha$-smooth, i.e.~where $\Pd^{R,x}(y) \leq \alpha 2^{-m}$ for all $y$ and some threshold parameter $\alpha = \poly(n)$ (we say all samples are $\alpha$-light). If the verifier knew the probability $\gYs:=\Pr_{y\leftarrow \Pd^{R,x}}[y \in L']$, it is easy to see that the Feigenbaum-Fortnow protocol (using $\gYs$ instead of $\gUYs$ as above) can be used to simulate the reduction. Unfortunately, $\gYs$ cannot be handed to the verifier as advice, as it may depend on the instance $x$. Thus, \cite{BT06} give a protocol, named the \textit{Hiding Protocol}, that allows the verifier to obtain an approximation of $\gYs$, given $\gUYs$ as advice. 

The idea of the protocol is as follows: the verifier hides a $1/\alpha$-fraction of samples from $\Pd^{R,x}$ among uniform random samples (i.e.~it permutes all samples randomly). The honest prover again indicates the yes-instances and provides witnesses for them. The verifier checks the witnesses and that the fraction of yes-answers among the uniform samples is approximately $\gUYs$. If this is true, it uses the fraction of yes-answers among the samples from $\Pd^{R,x}$ as an approximation of $\gYs$. 

Completeness follows easily. The intuition to see soundness is that since the distribution is $\alpha$-smooth, and as the verifier hides only a $1/\alpha$ fraction of $\Pd^{R,x}$ samples among the uniform ones, the prover cannot distinguish them. 

We note that the intuition behind this protocol crucially relies on the fact that the verifier can keep some of its random coins private: the prover is not allowed to know where the distribution samples are hidden.

\paragraph{General distributions and the Heavy Samples Protocol.}
Finally, \cite{BT06} remove the restriction that $\Pd^{R,x}$ is $\alpha$-smooth as follows. We say $y$ is $\alpha$-heavy if $\Pd^{R,x}(y) \geq \alpha 2^{-m}$, and let 
$\gHs := \Pr_{y\inu \Pd^{R,x}}[\Pd^{R,x}(y) \geq \alpha 2^{-m}]$ be the probability of a distribution sample being heavy, and $\gYLs := \Pr_{y\inu \Pd^{R,x}}[y\in L' \land \Pd^{R,x}(y) < \alpha 2^{-m}]$ the probability of a distribution sample being a yes-instance and light.

We first note that if the verifier knows (an approximation of) both $\gHs$ and $\gYLs$, it can use the Feigenbaum-Fortnow approach to simulate the reduction: the verifier simply uses $\gYLs$ instead of $\gUYs$ in the protocol, and ignores the heavy samples. It can do this by having the prover indicate the $\alpha$-heavy instances, and checking that their fraction is close to $\gHs$. Using the lower bound protocol of Goldwasser and Sipser \cite{GS86} (see Section~\ref{sec:VerifyHistPLB}), the prover must prove that these samples are indeed heavy. Finally, for the heavy samples the verifier can simply set the oracle answers to $0$: this changes the oracle answers on at most a polynomially small (i.e.~a $1/\alpha$) fraction of the inputs, as by definition at most a $1/\alpha$ fraction of the $y$'s can be $\alpha$-heavy. Completeness is not hard to see, and soundness follows because a cheating prover cannot claim light samples to be heavy (by the soundness of the lower bound protocol), and thus, by the verifier's check, cannot lie much about which samples are heavy. 

If the verifier knows (an approximation of) $\gHs$, then it can use the hiding protocol to approximate $\gYLs$: the verifier simply ignores the heavy samples. This is again done by having the prover additionally tell which samples are $\alpha$-heavy (and prove this fact using the lower bound protocol). The verifier additionally checks that the fraction of heavy samples among the distribution samples is approximately $\gHs$, and finally uses the fraction of light distribution samples as approximation for $\gYLs$. 

It only remains to approximate $\gHs$. This is done using the \textit{Heavy Samples Protocol} as follows: the verifier samples $y_1, \ldots, y_k$ from $\Pd^{R,x}$ by choosing random $r_1, \ldots, r_k$ and letting $y_i := R(x,r_i)$. It sends the $y_i$ to the prover. The honest prover indicates which of them are heavy, and proves to the verifier using the lower bound protocol of \cite{GS86} that the heavy samples are indeed heavy and using the upper bound protocol of Aiello and H{\aa}stad \cite{AH91} that the light samples are indeed light. 
The verifier then uses the fraction of heavy samples as its approximation for $\gHs$. It is intuitive that this protocol is complete and sound. 
The upper bound protocol requires that the verifier knows a uniform random element (which is unknown to the prover) in the set on which the upper bound is proved. In our case, the verifier indeed knows the value $r_i$, which satisfies this condition. 

We note that this protocol relies on private-coins, as the verifier must keep the $r_i$ secret for the upper bound proofs. 

\subsection{Our Proof}
We give two new protocols to approximate the probabilities $\gHs$ and $\gYLs$, as defined in the previous section. These protocols can be used to replace the Hiding Protocol and the Heavy Samples Protocol of \cite{BT06}, respectively. Together with the Feigenbaum-Fortnow based simulation protocol of \cite{BT06}, this then yields a different proof of $\co\NP \subseteq \AMpoly$ under the given assumptions. 

\paragraph{Verifying histograms.}
We are going to employ the VerifyHist protocol by Haitner et al.~\cite{HMX10} to verify the histogram of a probability distribution. Recall that the $(\eps,t)$-histogram $h = (h_0, \ldots, h_t)$ of a distribution $\Pd$ is defined by letting $h_i := \Pr_{y\leftarrow \Pd}[y \in \cB_i]$, where $\cB_i := \left\{x: \Pd(x) \in (2^{-(i+1)\eps}, 2^{-i\eps}]\right\}$ (See Definition~\ref{def:histogram}). We will use the VerifyHist protocol for the distribution $\Pd^{R,x}$, as defined by the reduction $R(x,\cdot)$ under consideration, i.e.~$\Pd^{R,x}(y)=\Pr_{r}[R(x,r)=y]$. Intuitively, this protocol allows to prove that some given histogram $h$ is close to the true histogram of $\Pd^{R,x}$ in terms of the $1$st Wasserstein distance (also known as Earth Mover's distance). This distance between $h$ and $h'$ measures the minimal amount of work that is needed to push the configuration of earth given by $h$ to get the configuration given by $h'$: moving earth over a large distance is more expensive than moving it over a short distance. For formal definitions of histograms and the $1$st Wasserstein distance we refer to Section~\ref{sec:HistAndWasserstein}. 
\begin{lemma}[VerifyHist protocol of~\cite{HMX10}, informal]
	There is a constant-round public-coin protocol VerifyHist
	where the prover and the verifier get as input the circuit
	$R(x,\cdot)$ and a histogram $h$, and we have:
	
	\vspace{0.2cm}
	\noindent\textbf{Completeness:}
		If $h$ is the histogram of $\Pd^{R,x}$, then the verifier
		accepts with high probability.
		
	\vspace{0.2cm}	
	\noindent\textbf{Soundness:} If $h$ is far from the
		histogram of $\Pd^{R,x}$ in the $1$st Wasserstein 
		distance, then the verifier rejects with high	probability. 
\end{lemma}
The formal statement can be found in Section~\ref{sec:VerifyHistPLB}.

\paragraph{The new Heavy Samples Protocol.}
The idea to approximate the probability $\gHs$ is very simple. The honest prover sends the histogram of $\Pd^{R,x}$, and the verifier uses the VerifyHist protocol to verify it. Finally, the verifier simply reads the probability $\gHs$ from the histogram. 

There is a technical issue that comes with this approach. For example, it may be that all $y$'s with nonzero probability have the property that $\Pd^{R,x}(y)$ is very close, but just below $\alpha 2^{-m}$. In this case, a cheating prover can send a histogram claiming that these $y$'s have probability slightly above this threshold. This histogram has small Wasserstein distance from the true histogram, as the probability mass is moved only over a short distance. Clearly, the verifier's guess for $\gHs$ is very far from the true value in this case. 

We note that the same issue appears in the proof of \cite{BT06}, and we deal with it in exactly the same way as they do: we choose the threshold $\alpha$ randomly, such that with high probability $\Pr_{y\leftarrow \Pd^{R,x}}[\Pd^{R,x}(y) \text{ is close to }\alpha 2^{-m}]$ is small (see Section~\ref{sec:RandomThreshold} for the formal statement). 

\paragraph{A public-coin Hiding Protocol for smooth distributions.}
We would like the verifier to only send uniform random samples to the prover (as opposed to the original hiding protocol, where a few samples from the distribution are hidden among uniform samples). We first describe the main idea in the special and simpler case where $\Pd^{R,x}$ is $\alpha$-smooth. In this case, we can give the following protocol, which uses $\gUYs$ as advice: 

The verifier sends uniform random samples $y_1, \ldots, y_k$. The prover indicates for each sample whether it is a yes-instance, and provides witnesses. Furthermore, the prover tells $\Pd^{R,x}(y_i)$ to the verifier, and proves a lower bound on this probability. The verifier checks the witnesses and if the fraction of yes-instances is approximately $\gUYs$, and considers the histogram $h$ induced by the probabilities $\Pd^{R,x}(y_i)$, and in particular checks if the probability mass of $h$ is $1$. Finally, the verifier considers the histogram $h_Y$ induced by only considering the yes-instances, and uses the total mass in $h_Y$ as its approximation of $\gYLs$.

To see completeness, the crucial point is that the smoothness assumption implies that the verifier can get a good approximation of the true histogram.

Soundness follows because the prover cannot claim the probabilities to be too large (as otherwise the lower bound protocol rejects), and it cannot claim many probabilities to be too small, as otherwise the mass of $h$ gets significantly smaller than $1$. As it cannot lie much about yes-instances, this implies a good approximation of $\gYLs$. 

\paragraph{Dealing with general distributions.}
The above idea can be applied even to general distributions, assuming that the verifier knows the probability 
$\gUHs:= \Pr_{y\inu \{0,1\}^m}[\Pd^{R,x}(y) 
															 \geq \alpha 2^{-m}]$ 
of a uniform random sample being heavy. The prover still provides the same information. The verifier only considers the part of the induced histogram $h$ below the $\alpha2^{-m}$ threshold, and checks that the mass of $h$ below the threshold is close to $1-\gUHs$. 

As in the heavy samples protocol, we again encounter the technical issue that many $y$'s could have probability close to the threshold, in which case the prover can cheat. But, as discussed earlier, this situation occurs with small probability over the choice of $\alpha$. 

\paragraph{Approximating the probability of a uniform sample being heavy.}
Thus, it remains to give a protocol to approximate $\gUHs$. We do this in exactly the same way as the Heavy Samples protocol approximates $\gHs$. That is, given the histogram that was verified using VerifyHist, the verifier simply reads the approximation of $\gUHs$ from the histogram. The proof that this works is rather technical, as we must show that small Wasserstein distance between the true and the claimed histogram implies a small difference of the probability $\gUHs$ and its approximation read from the claimed histogram. We note that we include the protocol for approximating $\gUHs$ directly into our Heavy Samples protocol. 

\section{The New Protocols}
\label{sec:AvgCaseNewProtocols}
We give protocols to replace the Heavy Samples Protocol and the Hiding Protocol of \cite{BT06}. A technical overview including proof intuitions can be found in Section~\ref{sec:TechnicalOverviewWTA}. 
In this section, we give the two protocols and state the guarantees they give. The protocol analyses can be found in Sections~\ref{sec:AnalysisHeavySamples} and \ref{sec:AnalysisHiding}. 

\subsection{Choosing a Random Threshold}
\label{sec:RandomThreshold}
We let $\cA_{\alpha_0,\delta}$ be the uniform distribution on $\{\alpha_0(1+3\delta)^i: 0\leq i \leq 1/\delta\}$. This distribution will be used to choose a 
threshold parameter $\alpha$. The following claim is from \cite{BT06}. 

\begin{claim}[Choosing a random threshold]
	\label{claim:threshold}
	For every $\alpha_0 >0$ and $0<\delta<1/3$, and every distribution
	$\Pd$ on $\{0,1\}^m$ we have
	$$ \Exp_{\alpha \leftarrow \cA_{\alpha_0,\delta}}\left[
			\Pr_{y \leftarrow \Pd}[\Pd(y) \in (1 \pm \delta) \alpha 2^{-m}]
		 \right] \leq \delta.$$
\end{claim}

We get that with high probability over the choice of $\alpha$ there is only little mass close to the threshold: 
\begin{claim}
	\label{claim:smallMassAroundThreshold}
	For every distribution
	$\Pd$ on $\{0,1\}^m$ and $\eps \in (0,1)$,
	with probability at least $1-20\sqrt{\eps}$
	over the choice of $\alpha$ from $\cA_{\alpha_0, 4\eps}$,
	we have
	$ 
	  \Pr_{y \leftarrow \Pd}[\Pd(y) \in (1 \pm 4\sqrt{\eps}) 
				\alpha 2^{-m}] \leq \frac{1}{5}\sqrt{\eps}$.
\end{claim}
\begin{proof}
	This follows from Claim~\ref{claim:threshold} by applying
	Markov's inequality.
\end{proof}

\subsection{Preliminaries}
Since our protocols can be used to replace part of the proof of  \cite{BT06}, we mostly stick to their notation. 
We give a formal definition of interactive proofs, histograms, and the Wasserstein distance in Section~\ref{sec:SamplingPreliminaries}. As in \cite{BT06}, we let $\cA_{\alpha_0,\delta}$ be the uniform distribution on $\{\alpha_0(1+3\delta)^i: 0\leq i \leq 1/\delta\}$. This distribution will be used to choose a threshold parameter $\alpha$, such that only little probability mass is close to the threshold (see Section~\ref{sec:RandomThreshold}). We consider circuits $C: \{0,1\}^n \rightarrow \{0,1\}^m$, and the distribution $\Pd^C$ defined by $\Pd^C(y) = \Pr_{r\leftarrow \{0,1\}^n}[C(r) =y]$. We use the VerifyHist protocol of \cite{HMX10} which on input $(C,\eps,h)$ verifies that $h$ is close to the $(t,\eps)$-histogram of $\Pd^C$ (see Section~\ref{sec:VerifyHistPLB}). We also use the parallel lower bound protocol as stated in \cite{BT06}, which on input $(C,\eps,y_1, s_1, \ldots, y_k,s_k)$ ensures that $\forall i: |C^{-1}(y_i)| \geq (1-\eps)s_i$ (see Section~\ref{sec:VerifyHistPLB}). 

We define the following probabilities. 
For a given threshold parameter $\alpha > 1$, a circuit $C:\{0,1\}^n\rightarrow \{0,1\}^m$ and a nondeterministic circuit\footnote
{A \textit{nondeterministic circuit} $V$ is of the form 
	$V: \{0,1\}^m\times \{0,1\}^{\ell} \rightarrow \{0,1\}$, and 
	we say $y \in \{0,1\}^m$ is accepted by $V$ 
	(or also $y \in V$) if there exists $w \in \{0,1\}^{\ell}$
	such that $V(y,w)=1$.
} 
$V:\{0,1\}^m \times \{0,1\}^{\ell} \rightarrow \{0,1\}$ we let
\begin{align*}
	\begin{array}{ll}\vspace{.1cm}
					\gH = \Pr_{y \leftarrow \Pd^C}
						[\Pd^C(y) \geq \alpha 2^{-m}], &
					\gUH = \Pr_{y \leftarrow \{0,1\}^m}
							[\Pd^C(y) \geq \alpha 2^{-m}],	\\ 
					\gUY = \Pr_{y \leftarrow \{0,1\}^m}[y \in V], & 
		\gYL = \Pr_{y \leftarrow \Pd^C}[\Pd^C(y) < \alpha 
							2^{-m} \land y\in V].\\
	\end{array}
\end{align*}

In the technical overview as given above (Section~\ref{sec:TechnicalOverviewWTA}), we considered the circuit $R(x,r)$ defined by the reduction $R$, which for a fixed $x$ and randomness $r$ outputs the first reduction query $y$. The protocols we give in the following are supposed to get as input the circuit $C(r):=R(x,r)$ for a fixed $x$. Furthermore, for the new hiding protocol the input circuit $V$ is supposed to provide an $\NP$ verifier for the language $L'$ as described in the technical overview. 

\subsection{The new Heavy Samples Protocol}

Given a circuit $C$ and $\alpha > 0$, the goal of the heavy 
samples protocol is to estimate the probability of heavy elements. The first probability, which we denote by $\gH$ is the probability that an element chosen from $\Pd^C$ is $\alpha$-heavy, i.e.~satisfies $\Pd^C(y) \geq \alpha 2^{-m}$. 
The second probability, denoted $\gUH$ is the probability
that a uniform random element satisfies this property. 

\newcommand{\PiHeavy}{\ensuremath{\Pi^{\text{PubHeavy}, \alpha}}}
We give a protocol for the family of promise problems $\{\PiHeavy\}$, which is defined as follows.
\begin{align*}
	\PiHeavy_Y &:= \left\{ 
			(C, \pH, \pUH, \varepsilon) 
			: \pH = \gH \land \pUH = \gUH
		\right\} \\
	\PiHeavy_N &:= \left\{
		  (C, \pH, \pUH, \varepsilon)
			: \pH \notin [\gH\pm \frac{4}{5}\sqrt{\varepsilon}] \lor 
				\pUH \notin [\gUH\pm 10\sqrt{\varepsilon}]
	\right\}
\end{align*}
We assume that the input $(C, \pH, \pUH, \varepsilon)$ 
is such that $C: \{0,1\}^n \rightarrow \{0,1\}^m$ is a circuit,
$\pH,\pUH \in [0,1]$, and $\eps \in (0,1)$. 
The proof of the following theorem can be found in Section~\ref{sec:AnalysisHeavySamples}, and the protocol is stated below.

\begin{theorem}
	For every integer $\alpha_0$, with probability at least 
	$1-20\sqrt{\varepsilon}$ over the choice of $\alpha$ from 
	$\cA_{\alpha_0,4\varepsilon}$, the heavy samples protocol is
	a constant-round interactive proof for $\PiHeavy$ with
	completeness $1-2^{-n}$	and	soundness $1-2^{-n}$, where the
	verifier runs in time	$\poly(\frac{\Size(C)}{\eps})$. 
\end{theorem}

\noindent\textbf{The heavy samples protocol.} On input $(C,\pH,\pUH,\eps)$:
\begin{itemize}			
	\item[] \textbf{Prover:} 
					Let $t := \left\lfloor \frac{n}{\teps}\right\rfloor$
					and
					$\teps := (\frac{4}{100})^2 \varepsilon^2$, and
					send an $(\teps,t)$-histogram $h\in [0,1]^{t+1}$ 
					to the verifier.
					
					If the prover is honest, it sends the
					$(\teps,t)$-histogram of $\Pd^C$, denoted by $h^C$.
	\item[] \textbf{Prover and Verifier:} 
					Run the VerifyHist protocol 
					(Lemma~\ref{lem:VerifyHist}) on input 
					$(C,\teps, h)$. The verifier rejects in case
					that protocol rejects.
	\item[] \textbf{Verifier:} 
					Let $j^* := 
					\max\{j: 2^{-(j+1)\teps} > \alpha 2^{-m}\}$. 
					Accept if and only if all of the following conditions hold:
					\begin{align*}
						\text{(a)}\qquad & 
							\sum_{j \in 
										\{j^*\pm \lceil 25/\sqrt{\teps}\rceil\}}
										h_j \leq \teps^{1/4}
						\qquad\qquad\quad
						\text{(b)}\qquad  
							\sum_{j \leq j^*} 
									h_j
							\in [\pH \pm \teps^{1/4}] \\
						\text{(c)}\qquad &
							\frac{1}{2^m}\sum_{j \leq j^*}
									h_j \cdot 2^{j\teps}
							\in [\pUH \pm 4\teps^{1/4}]
					\end{align*} 
\end{itemize}			

\subsection{The new Hiding Protocol}

Given a circuit $C$, a nondeterministic circuit $V$, and $\alpha > 0$, the goal of the hiding protocol is as follows. 
Given advice $\gUY$ and approximations of the probabilities $\gH$ and $\gUH$, the protocol approximates the probability $\gYL$ that an element is a yes-instance and $\alpha$-light. 

\newcommand{\PiPubHide}{\ensuremath{\Pi^{\text{Hide}, \alpha}}}
We give a protocol for the family of promise problems 
$\{\PiPubHide\}$, which is defined as follows.
\begin{align*}
	\PiPubHide_Y &:= \left\{ 
			(C, V, \pH, \pUH, \pYL, \varepsilon) 
			: \pH = \gH \land \pUH = \gUH \land \pYL = \gYL
		\right\} \\
	\PiPubHide_N &:= \Bigl\{
		  (C, V, \pH, \pUH, \pYL, \varepsilon) 
			: \pH \in [\gH \pm \frac{4}{5}\sqrt{\varepsilon}] \land 
				\pUH \in [\gUH \pm 10\sqrt{\varepsilon}]  \\
				&\qquad\qquad\qquad\qquad\qquad\qquad\land
				\pYL \notin [\gYL\pm 117\sqrt{\varepsilon}\alpha]
	\Bigr\}
\end{align*}
We assume that the input 
$(C, 
V, 
\pH, 
\pUH, 
\pYL, 
\varepsilon) $ 
is such that $C: \{0,1\}^n \rightarrow \{0,1\}^m$ is a circuit,
$V: \{0,1\}^m \times \{0,1\}^{\ell} \rightarrow \{0,1\}$ is a
nondeterministic circuit, $\pH,\pUH, \pYL \in [0,1]$, and 
$\eps \in (0,1)$. 
The proof of the following theorem can be found in Section~\ref{sec:AnalysisHiding}, and the protocol is given below.

\begin{theorem}
	For every integer $\alpha_0$, with probability at least 
	$1-20\sqrt{\varepsilon}$ over the choice of $\alpha$ from 
	$\cA_{\alpha_0,4\varepsilon}$,	
	the hiding protocol with advice $\gUY$ is a
	constant-round interactive 
	proof for $\PiPubHide$ with completeness $1-5\varepsilon$
	and soundness $6\varepsilon$, where the verifier runs in 
	time $\poly(\frac{\Size(C)+\Size(V)}{\eps})$. 
\end{theorem}

\noindent\textbf{The hiding protocol.} On input $(C, V, \pH, \pUH, \pYL, \varepsilon)$ and advice $\gUY$: 
\newcommand{\cJh}{\mathcal{J}_{\text{heavy}}}
\newcommand{\cAl}{\mathcal{A}_{\text{light}}}
\newcommand{\cBh}{\mathcal{B}_{\text{heavy}}} 
\begin{itemize}			
	\item[] \textbf{Verifier:} 
					Let $t:=\left\lceil \frac{n}{\varepsilon}\right\rceil$
					and let	$\cB_i$ for $i\in (t)$ be defined as in 
					Definition~\ref{def:histogram}.
					Let
					$k:=\ln(\frac{2}{\varepsilon})\alpha^2
							\frac{9}{2\varepsilon^2}$.
					Choose $y_1, \ldots, y_{k} \leftarrow \{0,1\}^m$, and 
					send $y_1, \ldots, y_{k}$ to the prover.
	\item[] \textbf{Prover:} 
					Send a labeling $u$, a set $\cY\subseteq [k]$ and 
					witnesses $(w_i)_{i\in \cY}$ to the verifier. 
					
					If the prover is honest, it sends 
					$\cY:=\{i: y_i \in V\}$, 
					and witnesses $(w_i)_{i\in \cY}$ such that 
					$V(y_i,w_i)=1$, and\footnote
					{For the special symbol $\infty$, we use the
						conventions $2^{-\infty \varepsilon} = 0$, 
						$\forall i \in \mathbb{N}: i < \infty$, 
						and 
						$\forall i \in \mathbb{N}: \infty + i = \infty$.
					}  $u$ such that for $i\in [k]$ we have
						$
									u(i) = 
									\begin{cases}
										j     & \text{if } \exists j: 
																				y_i \in \cB_j \\
										\infty  & \text{otherwise.} 
									\end{cases}
						$
	\item[] \textbf{Verifier:} 
					Let $\cL := \left\{i: 
												2^{-(u(i)+1) \varepsilon} <
												\alpha 2^{-m}\right\}$, 
							$\cH := [k]\setminus \cL$, and
					reject if one of the following 
					conditions does not hold:
					
					\begin{align*}
						\text{(a)} &\quad 
							\frac{|\cY|}{k} \in [\gUY\pm \varepsilon],
						\hspace{0.9cm}
						\text{(b)} \quad
							\forall i\in \cY: V(y_i, w_i) = 1, \\
						\text{(c)} &\quad
							\frac{|\cH|}{k}
									\in
									[\pUH \pm 3\sqrt{\varepsilon}],
						\qquad
						\text{(d)} \quad
							\frac{1}{k} \sum_{i\in \cL}
									2^m\cdot 2^{-u(i)\varepsilon}\in [
									1-\pH \pm 5\sqrt{\varepsilon}],\\
						\text{(e)} &\quad
							\frac{1}{k} \sum_{i \in \cL \cap \cY}
									2^m\cdot 2^{-u(i)\varepsilon}\in [
									\pYL \pm 5\sqrt{\varepsilon}].
					\end{align*}
	\item[] \textbf{Prover and Verifier:} 
					Run the parallel lower bound protocol 
					(see Lemma~\ref{lem:PLBprotocol})	on input
					$(C, \eps/2, y_1, s_1, \ldots, y_k, s_k)$, using
					the values	$s_i = 2^m \cdot 
					2^{-(u(i)+1)\varepsilon}$. 
\end{itemize}

\section{Analysis of the New Heavy Samples Protocol}
\label{sec:AnalysisHeavySamples}

Throughout the proof, we will use $\teps := (\frac{4}{100})^2 \varepsilon^2$. Note that then 
\begin{align*}
	\PiHeavy_N = \left\{
		  (C, \pH, \pUH, \varepsilon)
			: \pH \notin [\gH\pm 4\teps^{1/4}] \lor 
				\pUH \notin [\gUH\pm 50\teps^{1/4}]
	\right\},
\end{align*}
Also, Claim~\ref{claim:smallMassAroundThreshold} states the following when substituting $\teps$ for $\eps$:
\begin{claim}
	\label{claim:smallMassAroundThresholdTeps}
	For every distribution
	$\Pd$ on $\{0,1\}^m$ and $\eps \in (0,1)$,
	with probability at least $1-100\teps^{1/4}$
	over the choice of $\alpha$ from $\cA_{\alpha_0, 4\eps}$,
	we have
	$
	  \Pr_{y \leftarrow \Pd}[\Pd(y) \in (1 \pm 100\sqrt{\teps}) 
				\alpha 2^{-m}] \leq \teps^{1/4} $.
\end{claim}

\subsection{Proof of Completeness: Overview}
\label{sec:HeavySamplesCompletenessOverview}
We use the following lemma, which states that if there is only
little mass around the threshold $\alpha 2^{-m}$, then the 
verifier's checks are indeed satisfied for the honest prover
who sends $h^C$.
\begin{lemma}
	\label{lem:propertiesOfHC}
	Suppose 
	\begin{align}
	  \Pr_{y \leftarrow \Pd^C}[\Pd^C(y) \in 
		(1 \pm 100\sqrt{\teps}) \alpha 2^{-m}] \leq \teps^{1/4}.
		\label{eqnp:a1}
	\end{align}
	Then we have 
	\begin{align*}
		\text{(i)} \qquad &
		2^{-(j^*-\lceil 25/\sqrt{\teps} \rceil)\teps} \leq 2^{28\sqrt{\teps}} \alpha 2^{-m}\\
		&2^{-(j^*+\lceil 25/\sqrt{\teps} \rceil)\teps} \geq 2^{-28\sqrt{\teps}} \alpha 2^{-m}\\
		\text{(ii)} \qquad &
		\sum_{j \in 
										\{j^*\pm \lceil 25/\sqrt{\teps}\rceil\}}
										h^C_j \leq \teps^{1/4}\\
		\text{(iii)} \qquad &
		\sum_{j\leq j^*} 
									h^C_j
							\in [\gH \pm \teps^{1/4}] \\
		\text{(iv)} \qquad &
		\frac{1}{2^m}\sum_{j \leq j^*}
									h^C_j \cdot 2^{j\teps}
							\in [\gUH \pm 4\teps^{1/4}]
	\end{align*}
\end{lemma}
With this, it is straightforward to prove completeness:
\begin{proof}[Proof of completeness]
	With high probability there is indeed little mass around
	the threshold: by 
	Claim~\ref{claim:smallMassAroundThresholdTeps}, 
	(\ref{eqnp:a1}) holds with probability at least 
	$1-100\teps^{1/4} = 1-20\sqrt{\eps}$ 
	over the choice of $\alpha$.
	Furthermore, by the completeness of VerifyHist, that 
	protocol accepts the true histogram $h^C$ with probability
	at least $1-2^{-n}$. Finally, the above lemma gives that 
	(\ref{eqnp:a1}) implies	(a)-(c).
\end{proof}

It remains to prove Lemma~\ref{lem:propertiesOfHC}. We focus on
the interesting parts of the proof, and defer the technical details to 
Section~\ref{sec:HeavySamplesComplDetails}.

\begin{proof}[Proof of Lemma~\ref{lem:propertiesOfHC}]
	We defer the proofs of (i) and (ii) to Section~\ref{sec:HeavySamplesComplDetails}.  
	Part (i) is a straightforward calculation, which follows
	by the definition of $j^*$. Part (ii) then follows from 
	part (i), since the probabilities we sum over are close to
	the threshold $\alpha 2^{-m}$ and we can apply 
	(\ref{eqnp:a1}). 
	
	Now part (iii) is easy to prove: by definition of $j^*$ and
	$h_j^C$ we have
	\begin{align*}
		\gH \in \bigl[\sum_{j\leq j^*}h_j^C, \sum_{j\leq j^*+1} 
			h_j^C
		\bigr].
	\end{align*}
	Now (ii) gives $h^C_{j^*+1}\leq \teps^{1/4}$, which gives the
	claim.
	
	The proof of (iv) again follows using (ii) (i.e.~as there is only little mass close to the threshold), 
	and we give the proof in Section~\ref{sec:HeavySamplesComplDetails}. 
\end{proof}

\subsection{Proof of Soundness: Overview}
We use the following lemma, which states that if there is 
only little mass around the threshold, and the guarantee 
on the Wasserstein distance (which holds with high probability
by the soundness of VerifyHist) indeed holds, then the 
values the verifier computes are close to the true values
$\gH$ and $\gUH$. 
\begin{lemma}
	\label{lem:soundnessHeavySamples}
	Suppose that the verifier's check (a) holds,
	\begin{align}
	  \Pr_{y \leftarrow \Pd^C}[\Pd^C(y) \in 
		(1 \pm 100\sqrt{\teps}) \alpha 2^{-m}] \leq \teps^{1/4},
		\label{eqnp:a2}
	\end{align}
	and $\Wd(h^C, h) \leq \frac{20}{t}$. 
	Then we have
	\begin{align*}
		\text{(i)} \qquad &
			\sum_{j\leq j^*} h_j
			\in \bigl[
				\sum_{j\leq j^*} h^C_j \pm 
				2 \teps^{1/4}
			\bigr],\\
		\text{(ii)} \qquad &
			\sum_{j\leq j^*} h_j
			\in \bigl[
				\gH \pm 
				3 \teps^{1/4}
			\bigr],\\
		\text{(iii)} \qquad &
			\sum_{j\leq j^*} h_j 2^{j\teps} \
				\in \bigl[
					\gUH \pm 46 \teps^{1/4}
				\bigr].
	\end{align*}
\end{lemma}
With this lemma, it is straightforward to prove soundness:
\begin{proof}[Proof of soundness]
	By Claim~\ref{claim:smallMassAroundThresholdTeps}, 
	(\ref{eqnp:a2}) holds with probability at least 
	$1-100\teps^{1/4}= 1-20\sqrt{\eps}$ over the choice of 
	$\alpha$.
	Now, by the soundness of VerifyHist, we get that
	$\Wd(h^C, h) \leq \frac{20}{t}$ with probability
	at least $1-2^{-n}$ (or the verifier rejects).
	Clearly if (a) does not hold, the verifier rejects. 
	If (a) holds, then by the above lemma we  have (ii)
	and (iii), which as we are considering a no-instance 
	of $\PiHeavy$	gives that one of the following holds:
	\begin{align*}
		\sum_{j\leq j^*}h_j \notin [\pH \pm \teps^{1/4}],
		\qquad
		\sum_{j\leq j^*}h_j2^{j\teps} \notin
		[\pUH\pm 4\teps^{1/4}].
	\end{align*}
	Thus the verifier rejects in (b) or (c).
\end{proof}

It remains to prove Lemma~\ref{lem:soundnessHeavySamples}. 
We focus on the interesting parts of the proof, and defer
the details to Section~\ref{sec:HeavySamplesSoundnessDetails}.

\begin{proof}[Proof of Lemma~\ref{lem:soundnessHeavySamples}]
	We defer the proof of (i). It is not hard to see that if
	$\sum_{j\leq j^*}h^j$ is not in the desired interval, then
	$\Wd(h^C, h)$ is big: by (a) and 
	Lemma~\ref{lem:propertiesOfHC} (ii), only little mass can
	be around the threshold for both $h$ and $h^C$, and thus 
	a lot of mass must be moved from below to above the threshold,
	or vice versa.
	
	Part (ii) can then be proved easily: 
	by Lemma~\ref{lem:propertiesOfHC} (iii), the interval in
	(i) is contained in $[\gH \pm 3\teps^{1/4}]$.
	
	It remains to prove (iii). For $j\in (t)$ we consider the
	differences $d_j:=h_j-h^C_j$. 
	Then our assumption $\Wd(h^C,h) \leq 
	\frac{20}{t}$ gives
	\begin{align}
		\frac{20}{t} \geq \Wd(h^C,h)
		= \frac{1}{t}\sum_{i\in (t)} 
		\Bigl|
			\sum_{j\leq i}h_j - \sum_{j\leq i}h_j^C
		\Bigr|
		= \frac{1}{t}\sum_{i\in (t)} 
		\Bigl|
			\sum_{j\leq i}d_j
		\Bigr|
		= \Exp_{i \leftarrow (t)} \left[
			\Bigl|
			\sum_{j\leq i}d_j
		\Bigr|
		\right].
		\label{eqnp:45}
	\end{align}
	Furthermore, part (i) gives
	\begin{align}
		\sum_{j \leq j^*}d_j =
		\sum_{j \leq j^*}h_j - \sum_{j \leq j^*}h^C_j
		\stackrel{\text{(i)}}{\in} [\pm 2\teps^{1/4}].
		\label{eqnp:46}
	\end{align}
	At this point, we will use Lemma~\ref{lem:pUHdifferslittle}
	as stated below which contains the core of the argument.
	As (\ref{eqnp:45}) and (\ref{eqnp:46})
	hold, we may apply this lemma	to the $d_j$ as defined
	above, and obtain 
	\begin{align*}
		\frac{1}{2^m} \sum_{j\leq j^*}d_j 2^{j\teps}
		 \in \Bigl[\pm 42\teps^{1/4}\Bigr].
	\end{align*}
	Plugging in the definition of the $d_j$, we get
	\begin{align*}
		\frac{1}{2^m}\sum_{j\leq j^*}h_j 2^{j\teps}
		\in \Bigl[\frac{1}{2^m}\sum_{j\leq j^*}h^C_j 2^{j\teps} \pm 42\teps^{1/4}\Bigr]
		\subseteq \bigl[\gUH \pm 46\teps^{1/4}\bigr], 
	\end{align*}	
	where we used Lemma~\ref{lem:propertiesOfHC} (iv) for the above 
	set inclusion.
\end{proof}

\begin{lemma}
		\label{lem:pUHdifferslittle}
		Let $t$, $j^*$ and $\teps$ be as above, and fix any
		$d = (d_0, \ldots, d_{t}) \in \mathbb{R}^{t+1}$. 
		Suppose that 
		\begin{align*}
			\text{(i)} \qquad &	
				\sum_{j \leq j^*} d_j \in [-\delta, \delta], \qquad \qquad
			\text{(ii)} \qquad 
				\Exp_{i \leftarrow (t)}\bigl[ \bigl| \sum_{j\leq i} d_j		
				\bigr|\bigr] \leq \frac{20}{t}.
		\end{align*}
		Then we have
		$$
			\frac{1}{2^m} \sum_{j\leq j^*}d_j 2^{j\teps}
				\in \Bigl[\pm (\delta + 40\teps) \Bigr].
		$$
	\end{lemma}
	
	Again, we defer a few details of the proof to Section~\ref{sec:HeavySamplesSoundnessDetails}. 
	\begin{proof}
		We only prove the the inequality 
	$ \sum_{j\leq j^*}d_j 2^{j\teps}
				\leq (\delta + 40\teps)2^m$. The proof of
	$ \sum_{j\leq j^*}d_j 2^{j\teps}
				\geq -(\delta + 40\teps)2^m$ is analogous.	
	
	We first define a vector $d'$ such that
	for all $i < j^*$ it holds that 
	$\sum_{j\leq i} d_i' = \min\{\sum_{j\leq i}d_i, 0\}$, and 
	$\sum_{j\leq j^*} d_i'  = \sum_{j\leq j^*} d_i$. Note that this
	defines $d'$ uniquely. 	
	\begin{claim}
		\label{claim:dprime}
		We have
	\begin{align}
		&\Exp_{i\leftarrow (t)}\bigl[\bigl| 
				\sum_{j\leq i} d'_j
			\bigr|\bigr]
		\leq 
			\Exp_{i\leftarrow (t)}\bigl[\bigl| 
					\sum_{j\leq i} d_j
				\bigr|\bigr], \label{eqnp:37}\\
		&\sum_{j\leq j^*} d_j' 2^{j\teps} 
			\geq \sum_{j\leq i}d_j 2^{j\teps}.\label{eqnp:38}
	\end{align}
	\end{claim} 
	As the proof is not difficult, we defer it to 
	Section~\ref{sec:HeavySamplesSoundnessDetails}
	and just give some intuition here. The first part follows 
	by definition. To prove the second part, we show that
	$d'$ can be obtained from $d$ by moving mass from coordinate
	$i$ to coordinate $i+1$ for each $i$ individually. This then
	implies the claim, as moving mass to larger coordinates 
	only increases the sum.	
	
	Now define $d''$ as follows: $d''_j := d'_j$ for $j< j^*$,
$d''_{j^*} := d'_{j^*} - \sum_{j \leq j^*}d'_j$, and
$d''_j := 0$ for $j> j^*$. 
\begin{claim}
	\label{claim:ddoubleprime}
	We have
	\begin{align}
		&\sum_{j\leq j^*} d''_j = 0,
		\label{eqnp:42}
		\\	
		&\forall i \in (t): \sum_{j\leq i} d''_j \leq 0,
	\label{eqnp:41} \\
		&\Exp_{i\leftarrow (t)}\bigl[\bigl| 
				\sum_{j\leq i} d''_j
			\bigr|\bigr]
		\leq 
			\Exp_{i\leftarrow (t)}\bigl[\bigl| 
					\sum_{j\leq i} d'_j
				\bigr|\bigr], \label{eqnp:39}\\
		&\sum_{j\leq j^*} d''_j 2^{j\teps} 
			\geq \sum_{j\leq j^*}d'_j 2^{j\teps}
						- \delta 2^{j^*\teps}
			.\label{eqnp:40}
	\end{align}
\end{claim}
The proof is straightforward, and we defer it to 
Section~\ref{sec:HeavySamplesSoundnessDetails}.

\begin{claim}
	\label{claim:ddoubleprimeSumOfVectors}
	There exists $t \in \mathbb{N}$ 
	and vectors $v^{(1)}, \ldots, v^{(t)} \in 
	\mathbb{R}^{t+1}$ such that the following holds:
	\begin{enumerate}[(i)]
	\item $d'' = \sum_{a\in [t]} v^{(a)}$,
	\item For every $a \in [t]$, $v^{(a)}$ has 
				exactly two nonzero entries, whose index positions 
				we denote by $i(a)$ and $i'(a)$ where 
				$i(a) < i'(a) \leq j^*$. Furthermore, 
				$v^{(a)}_{i(a)} = -w_a$ and 
				$v^{(a)}_{i'(a)} = w_a$ for some 
				$w_a \in \mathbb{R}$, $w_a > 0$.
	\end{enumerate}
\end{claim}
We prove the claim in Section~\ref{sec:HeavySamplesSoundnessDetails}. There we show that 
the vectors $v^{(a)}$ can be defined iteratively 
by greedily picking the smallest nonzero index position $i$
(which must have negative $d_i$), and matching it with
the smallest index position $i'$ with $d_{i'} >0$. 

Now note that
\begin{align}
	\Exp_{i \leftarrow (t)}\bigl[
		\sum_{j\leq i} d''_j
	\bigr]
	&=
	\Exp_{i \leftarrow (t)}\bigl[
		\sum_{j\leq i} \sum_a v^{(a)}_j
	\bigr]
	= 
	\sum_a \Exp_{i \leftarrow (t)}\bigl[
		\sum_{j\leq i} v^{(a)}_j
	\bigr] \nonumber \\
	&= 
	-\sum_a \frac{1}{t} w_a (i'(a)- i(a)).
	\label{eqnp:43}
\end{align}
Then we find
\begin{align}
  \sum_{j\leq j^*}d''_j 2^{j\teps}
	&=
	\sum_{j\leq j^*}\sum_a v^{(a)}_j 2^{j\teps}
	=
	\sum_a \sum_{j\leq j^*} v^{(a)}_j 2^{j\teps}
	=
	\sum_a w_a (2^{i'(a) \teps} - 2^{i(a) \teps})\nonumber\\
	&=
	\sum_a w_a 2^{i'(a) \teps} (1 - 2^{(i(a)- i'(a)) \teps})\nonumber\\
	&\leq
	\sum_a w_a 2^{i'(a) \teps} (i'(a)-i(a))(1-2^{-\teps}) \nonumber\\
	&=
	(1-2^{-\teps}) \sum_a w_a 2^{i'(a) \teps} (i'(a)-i(a))\nonumber\\
	&\leq
	(1-2^{-\teps}) \sum_a w_a 2^{j^* \teps} (i'(a)-i(a))
	\nonumber\\
	&\stackrel{\text{(\ref{eqnp:43})}}{=}
	(1-2^{-\teps}) 2^{j^* \teps} \left(-t \cdot \Exp_{i \leftarrow (t)}\bigl[
		\sum_{j\leq i} d''_j
	\bigr]\right) \nonumber\\
	&\stackrel{\text{(\ref{eqnp:41})}}{=} 
	(1-2^{-\teps}) 2^{j^* \teps}\cdot t \cdot \Exp_{i \leftarrow (t)}\bigl[
		\bigl|\sum_{j\leq i} d''_j\bigr|
	\bigr] \nonumber\\
	&\stackrel{\text{(ii)}}{\leq} 2\teps \cdot 2^{j^* \teps} \cdot t \cdot \frac{20}{t}
	= 40 \teps 2^{j^* \teps}.
	\label{eqnp:44}
\end{align}
The first inequality above follows by Bernoulli's 
inequality\footnote{Bernoulli's inequality states that for any $n \in \mathbb{N}$, $n \geq 0$ and any $x \in \mathbb{R}$, $x \geq -1$ we have $(1+x)^n \geq 1+nx$.}
when setting $n=i'(a)-i(a)$ and $x=2^{-\teps}-1$.   
We now conclude the argument by calculating
\begin{align*}
		\sum_{j\leq j^*}d_j 2^{j\teps}
	&\stackrel{\text{(\ref{eqnp:38})}}{\leq}
		\sum_{j\leq j^*}d'_j 2^{j\teps}
	\stackrel{\text{(\ref{eqnp:40})}}{\leq}
		\sum_{j\leq j^*}d''_j 2^{j\teps} + \delta 2^{j^*\teps}
	\stackrel{\text{(\ref{eqnp:44})}}{\leq}
	  (40\teps + \delta)2^{j^*\teps} \\
		&\leq (40\teps + \delta)2^m.
\end{align*}	
	
\end{proof}

\subsection{Proof of Completeness: the Details}
\label{sec:HeavySamplesComplDetails}
In the following, we give the parts of the proof of 
Lemma~\ref{lem:propertiesOfHC} that we omitted in Section~\ref{sec:HeavySamplesCompletenessOverview}.
\begin{proof}[Proof of (i)]
	From the definition of $j^*$ we get that
	$2^{-(j^*+1)\teps} \in (\alpha 2^{-m}, 
	2^{\teps}\alpha 2^{-m}]$ (otherwise, $j^*$ would not
	be maximal).
	Thus, we find
	\begin{align*}
		2^{-(j^*-\lceil 25/\sqrt{\teps} \rceil)\teps}
		&\leq 2^{-(j^*-26/\sqrt{\teps})\teps}
		= 2^{-(j^*+1)\teps}
		  2^{(26/\sqrt{\teps}+1)\teps}
	  \leq 
		  2^{\teps}\alpha 2^{-m} 
		  2^{27\sqrt{\teps}} \\
		&\leq 
			\alpha 2^{-m} 2^{28\sqrt{\teps}}, \\
		2^{-(j^*+\lceil 25/\sqrt{\teps} \rceil)\teps}
		&\geq
		2^{-(j^*+26/\sqrt{\teps})\teps}
		= 2^{-(j^*+1)\teps}
		  2^{-(26/\sqrt{\teps}-1)\teps}
		\geq 
		  \alpha 2^{-m}
			2^{-27\sqrt{\teps}}.\qedhere
	\end{align*}
\end{proof}

\begin{proof}[Proof of (ii)]
	Note that 
	$h_j^C= \sum_{y: \Pd^C(y) \in (2^{-(j+1)\teps}, 2^{-j\teps}]} \Pd^C(y)$. Since $j\leq j^*+\lceil 25/\sqrt{\teps}\rceil$,
	we only sum over $y$ such that
	$$ 
		\Pd^C(y) \geq 2^{-(j^*+\lceil 25/\sqrt{\teps}\rceil+1)\teps} 
		= 2^{-\teps} 2^{-(j^*+\lceil 25/\sqrt{\teps}\rceil)\teps} \stackrel{\text{(i)}}{\geq} \alpha 2^{-m} 2^{-29\sqrt{\teps}}.
	$$
	On the other hand, because $j\geq 
	j^*-\lceil 25/\sqrt{\teps}\rceil$, for all $y$ we sum over,
	we have
	$$
	  \Pd^C(y) \leq 2^{-(j^*-\lceil 25/\sqrt{\teps}\rceil)\teps} \stackrel{\text{(i)}}{\leq} \alpha 2^{-m} 2^{28\sqrt{\teps}}.
	$$
	Thus, we conclude that
	\begin{align*}
		\sum_{j \in \{j^*\pm \lceil 25/\sqrt{\teps}\rceil\}}
										h^C_j
		\leq
		\sum_{y: \alpha 2^{-m} 2^{-29\sqrt{\teps}}\leq \Pd^C(y) \leq \alpha 2^{-m} 2^{28\sqrt{\teps}}} \Pd^C(y) \leq \teps^{1/4},
	\end{align*}
	where the last inequality holds because $[2^{-29\sqrt{\teps}}, 2^{28\sqrt{\teps}}] \subseteq (1\pm 100 \sqrt{\teps})$,
	and thus (\ref{eqnp:a1}) can be applied.
\end{proof}

\begin{proof}[Proof of (iv)]
	By definition of $h^C_j$, we have that for each $j$
	$$
		\Pr_{y\leftarrow \{0,1\}^m}\bigl[\Pd^C(y) \in (2^{-(j+1)\teps}, 
			2^{-j\teps}]\bigr] \in \bigl[\frac{1}{2^m} h^C_j 2^{j\teps}, \frac{1}{2^m} h^C_j 2^{(j+1)\teps}\bigr].
	$$
	Thus the definition of $j^*$ gives
	\begin{align}
		\gUH \in \bigl[
			\frac{1}{2^m}\sum_{j\leq j^*}h^C_j2^{j\teps}, 
			\frac{1}{2^m}\sum_{j\leq j^*+1}h^C_j2^{(j+1)\teps}
		\bigr]. 
		\label{eqnp:36}
	\end{align}
	Now we find
	\begin{align*}
		\sum_{j\leq j^*+1}h^C_j2^{(j+1)\teps}
		&= \sum_{j\leq j^*}h^C_j2^{(j+1)\teps} + 
			 h^C_{j^*+1}2^{(j^*+2)\teps} \\
		&\stackrel{\text{(ii)}}{\leq}
			 2^{\teps}\sum_{j\leq j^*}h^C_j2^{j\teps} + 
			 \teps^{1/4} 2^{\teps} 2^{(j^*+1)\teps}\\
		&<
		 (1+2\teps)\sum_{j\leq j^*}h^C_j2^{j\teps} + 
			 \frac{\teps^{1/4} 2^{\teps}}{\alpha}2^m \\
		&\leq
		\sum_{j\leq j^*}h^C_j2^{j\teps} + 
			 \bigl(2\teps + \frac{\teps^{1/4} 2^{\teps}}{\alpha}\bigr)2^m\\
		&\leq 
			\sum_{j\leq j^*}h^C_j2^{j\teps} + 4\teps^{1/4} 2^m,
	\end{align*}
	where the second inequality
	follows by definition of $j^*$, and the third inequality holds
	since $2^{j \teps} \leq 2^m$ for any $j \leq j^*$ and $\sum_{j\leq j^*} h_j^C \leq 1$. Plugging 
	this into (\ref{eqnp:36}) gives the claim.
\end{proof}

\subsection{Proof of Soundness: the Details}
\label{sec:HeavySamplesSoundnessDetails}

\begin{proof}[Proof of Lemma~\ref{lem:soundnessHeavySamples} (i)]
	First suppose
	\begin{align}
		\sum_{j\leq j^*} h_j
		< \sum_{j\leq j^*} h^C_j - 2\teps^{1/4}.
		\label{eqnp:a3}
	\end{align}
	We show that this implies
	$\Wd(h^C, h) > \frac{20}{t}$, contradicting our
	assumption.
	For any $i\in \{j^*-\lceil 25/\sqrt{\teps}\rceil, 
	\ldots, j^*\}$ we have
	\begin{align}
		\sum_{j \leq i} h_j 
		&\leq
		\sum_{j \leq j^*} h_j 
		\stackrel{\text{(\ref{eqnp:a3})}}{<} 
		\sum_{j \leq j^*} h^C_j - 2\teps^{1/4}
		\leq 
		\sum_{j \leq i} h^C_j - \teps^{1/4},
	\end{align}
	where the last inequality holds by 
	Lemma~\ref{lem:propertiesOfHC} (ii).
	This gives
	$\stackrel{\longleftarrow}{\Wd}(h^C,h) 
	 \geq \frac{1}{t} \cdot \frac{25}{\sqrt{\teps}} 
	      \cdot \teps^{1/4} \geq \frac{25}{t}$.
				
	Now assume that
	\begin{align}
		\sum_{j\leq j^*} h_j
		> \sum_{j\leq j^*} h^C_j + 2\teps^{1/4}.
		\label{eqnp:a4}
	\end{align}
	Again, we show that this implies
	$\Wd(h^C, h) > \frac{20}{t}$. Similar to above, 
	for any $i\in \{j^*-\lceil 25/\sqrt{\teps}\rceil, 
	\ldots, j^*\}$ we have
	\begin{align}
		\sum_{j \leq i} h_j 
		&\geq
		\sum_{j \leq j^*} h_j - \teps^{1/4}
		\stackrel{\text{(\ref{eqnp:a4})}}{>} 
		\sum_{j \leq j^*} h^C_j + \teps^{1/4}
		\geq
		\sum_{j \leq i} h^C_j + \teps^{1/4},
	\end{align}
	where the first inequality holds by the verifier's
	check (a). Thus
	$\stackrel{\longrightarrow}{\Wd}(h^C,h) 
	 \geq \frac{1}{t} \cdot \frac{25}{\sqrt{\teps}} 
	      \cdot \teps^{1/4} \geq \frac{25}{t}$.
\end{proof}

\begin{proof}[Proof of Claim~\ref{claim:dprime}]
	Inequality (\ref{eqnp:37}) holds because for each $i$, 
	$\bigl| \sum_{j\leq i} d'_i	\bigr| \leq \bigl| 
	\sum_{j\leq i} d_i \bigr|$ by definition. 
	
	To see (\ref{eqnp:38}), for each $k < j^*$ we define
	$e^{(k)} = (e_0^{(k)}, \ldots, e_{t}^{(k)})$ as follows.
	If $\sum_{j\leq k}d_j > 0$, we let 
	\begin{align*}
		e_i^{(k)} :=
		\begin{cases}
		 -\sum_{j\leq k}d_j & \text{if }  i = k, \\
		 \sum_{j\leq k}d_j  & \text{if }  i = k+1, \\
		 0       & \text{otherwise.} 
		\end{cases}
	\end{align*}
	and thus
	\begin{align*}
		\sum_{j\leq i} e_i^{(k)} =
		\begin{cases}
		 -\sum_{j\leq k}d_j & \text{if }  i = k, \\
		 0 & \text{otherwise.} 
		\end{cases}
	\end{align*}
		
	If $\sum_{j\leq k}d_j \leq 0$, we let $e_i^{(k)} = 0$ for all
	$i$. Now we find for any $k$ and $i$ that
	\begin{align}	
		\sum_{j\leq i} (d_j+e_j^{(k)}) = 
		\sum_{j\leq i} d_j + \sum_{j\leq i} e_j^{(k)} =
		\begin{cases}
		 \sum_{j\leq i}d_j & \text{if }  i \neq k, \\
		 \min\{\sum_{j\leq i}d_j,0\}     & \text{if }  i = k.  
		\end{cases}
	\end{align}
	This implies that	$d+ \sum_{k< j^*}e^{(k)} = d'$. Since 
	by definition it holds that
	$\sum_{j< j^*}e_j^{(k)} 2^{j\teps}\geq 0$ for any $k$, 
	we find
	\begin{align*}
		\sum_{j\leq j^*} d_j'2^{j\teps}
		&= \sum_{j\leq j^*} (d_j + \sum_{k<j^*}e_j^{(k)})2^{j\teps}
		= \sum_{j\leq j^*} d_j 2^{j\teps} 
			+\sum_{j\leq j^*} \sum_{k<j^*}e_j^{(k)} 2^{j\teps} \\
		&=\sum_{j\leq j^*} d_j 2^{j\teps} 
		 + \sum_{k<j^*} \sum_{j\leq j^*} e_j^{(k)} 2^{j\teps} 
	  \geq \sum_{j\leq j^*} d_j 2^{j\teps}. \qedhere
	\end{align*}	
\end{proof}

\begin{proof}[Proof of Claim~\ref{claim:ddoubleprime}]
	Equality (\ref{eqnp:42}) holds because 
	$\sum_{j\leq j^*} d''_j = \sum_{j< j^*} d'_j + d'_{j^*} - \sum_{j \leq j^*}d'_j = 0$,
	and (\ref{eqnp:41}) follows by definition.
	Inequality (\ref{eqnp:39}) follows because
	for $i < j^*$ we have 
	$\bigl| \sum_{j\leq i} d''_j \bigr| =
	 \bigl| \sum_{j\leq i} d'_j \bigr|$, and
	for $i\geq j^*$ we have
	$0=\bigl| \sum_{j\leq i} d''_j \bigr| \leq 
	 \bigl| \sum_{j\leq i} d'_j \bigr|$. 
	
	To see (\ref{eqnp:40}), we note that
	\begin{align*}
		\sum_{j\leq j^*}d_j''2^{j\teps} 
		&=
		\sum_{j<j^*}d_j'2^{j\teps} + 
		d'_{j^*}2^{j^*\teps} - (\sum_{j \leq j^*}d'_j) 2^{j^*\teps}
		= \sum_{j\leq j^*}d_j'2^{j\teps}
		  - (\sum_{j \leq j^*}d'_j) 2^{j^*\teps}\\
		&\stackrel{\text{(i)}}{\geq}
			\sum_{j\leq j^*}d_j'2^{j\teps} - \delta 2^{j^*\teps}.	\qedhere
	\end{align*}
\end{proof}

\begin{proof}[Proof of 
		Claim~\ref{claim:ddoubleprimeSumOfVectors}]
We define the vectors $v^{(1)}, \ldots, v^{(t)}$ using
the following procedure.
\lstdefinelanguage{pseudocode}{numbers=left,numberstyle=\tiny,mathescape,escapechar=*,morekeywords={System,abort,procedure,if,else,do,done,for,to,step,endif,forall,new,assert,end,return,while,then,private,public, Simulator, Variables,true,false},flexiblecolumns}
\begin{lstlisting}[language=pseudocode,name=findvectors]
  $a:=0$
  $f := d''$
  while $(\exists j: f_j \neq 0)$ do
    $a:=a+1$
    $i := \min\{j: f_j \neq 0\}$ 
    $i':= \min\{j: j>i \land f_{j}>0\}$
    $w := \min\{|f_i|, |f_{i'}|\}$
    for $j=0$ to $t$ do
      if $j=i$ then $v_j^{(a)}:= -w$
      else if $j=i'$ then $v_j^{(a)}:= w$
      else $v_j^{(a)}:= 0$
    $f:=f-v^{(a)}$
  $t:=a$
  return $(v^{(1)}, \ldots, v^{(t)})$
\end{lstlisting}

We claim that the following invariants always hold for $f$:
\begin{align*}
	\text{Invariant 1:} \quad \sum_{j\leq j^*}f_j = 0,
	\qquad \qquad
	\text{Invariant 2:} \quad \forall k\in (t): 
			\sum_{j\leq k}f_j \leq 0.
\end{align*}
By (\ref{eqnp:42}) and (\ref{eqnp:41}), the invariants
hold in the beginning where $f=d''$. Now suppose the invariants
hold for $f$ in some loop iteration, and we show they hold for $f'=f-v^{(a)}$ as defined in the next iteration, given $f$
still has a nonzero component. 
As invariant 2 holds for $f$, we have that $f_i < 0$, 
invariant 1 for $f$ implies that there exists $i'$ with
$f_{i'}>0$. The definition of $v_j^{(a)}$ directly implies that
invariant 1 holds for $f'$. Invariant 2 clearly holds
for $f'$ for any $k<i$, as the sum does not change. 
For $k=i$ it holds because $w < |f_i|$ and $f_i$ is the first 
non-zero component. For $i<k<i'$ it holds
because 
$$\sum_{j\leq k} f'_j \leq \sum_{j \leq i} f'_j \leq 0,$$
where the first inequality holds by the minimality of $i'$, 
and the second inequality is invariant 2 for $k=i$. 
Finally, the second invariant also holds for $k\geq i'$, as
then
$$ \sum_{j\leq k} f'_j 
  = \sum_{j\leq k} (f_j - v^{(a)}_j)
	= \sum_{j\leq k} f_j - 
		\underbrace{\sum_{j\leq k} v^{(a)}_j}_{= -w + w = 0} 
	\leq 0, $$
where we applied invariant 2 for $f$ to obtain the
inequality.	

Finally, in every iteration some nonzero component of $f$ (either $f'_i$ or $f'_{i'}$) is set to $0$. 
Thus the procedure terminates, and in the end we have
$\sum_i f_i = 0$ and $\sum_{a=1}^t v^{(a)} = d''$. Clearly,
the vectors $v^{(a)}$ satisfy (ii).
\end{proof}

\section{Analysis of the New Hiding Protocol}
\label{sec:AnalysisHiding}
Throughout this section, we let $u', \cY', \cL', \cH'$ be the
values as defined by the honest prover's strategy.
\subsection{Proof of Completeness: Overview}
We define the labeling $u'$ for all $y \in \{0,1\}^m$ as follows:
$$
	u'(y) = 
	\begin{cases}
		j     & \text{if } \exists j: 
												y \in \cB_j, \\
		\infty  & \text{otherwise.} 
	\end{cases}
$$
Note that the honest prover sends a labeling $u=u'$ such that
$u'(y_i)=u'(i)$. By definition, we have
\begin{claim}
	\label{claim:cDandUIntervals}
	$\Pd^C(y) \in (2^{-(u'(y)+1)\varepsilon},2^{-u'(y)\varepsilon}].$ 
\end{claim}

The following lemma states that if the prover is honest, then
the values calculated by the verifier in (a), (c)-(e) are close
to the true values as defined by $\Pd^C$. 
\begin{lemma}
	\label{lem:HighProbEstimates}
	Let $\cS:=\{y_1, \ldots, y_k\}$, 
	$\cM:=\{y: \Pd^C(y) \in (1\pm 4\varepsilon)\alpha 2^{-m}\}$
	, and assume
	\begin{align}
		\label{eqnp:5}
		\Pr_{y \leftarrow \Pd^C}[\Pd^C(y) \in 
							(1 \pm 4\varepsilon) \alpha 2^{-m}] 
		\leq \sqrt{\varepsilon},
	\end{align} 
	then
	\begin{align*}
		\text{(i)} \qquad 
			&\Pr_{y_1, \ldots, y_k}
			\left[\frac{|\cY'|}{k} \notin [\gUY \pm
						\varepsilon]\right]
			\leq \varepsilon \\
		\text{(ii)} \qquad 
			&\Pr_{y_1, \ldots, y_k}
			\left[\frac{|\cH'|}{k} \notin [\gUH \pm 
						3\sqrt{\varepsilon}]\right]
			\leq \varepsilon \\
		\text{(iii)} \qquad 
			&\Pr_{y_1, \ldots, y_k}
			\left[
				\frac{1}{k} \sum_{i\in \cL'}
									2^m\cdot 2^{-u'(i)\varepsilon}\notin [
									1-\gH \pm 5\sqrt{\varepsilon}]
			\right] \leq \varepsilon\\
		\text{(iv)} \qquad
		  &\Pr_{y_1, \ldots, y_k}
			\left[
			\frac{1}{k} \sum_{i \in \cL' \cap \cY'}
									2^m\cdot 2^{-u'(i)\varepsilon}\notin [
									\gYL \pm 5\sqrt{\varepsilon}]
		  \right] \leq \varepsilon\\
		\text{(v)} \qquad
		  &\Pr_{y_1, \ldots, y_k}
				\left[
					\frac{|\cS \cap \cM|}{k} 
					\geq \frac{3\sqrt{\varepsilon}}{\alpha}
				\right]
				\leq \varepsilon
	\end{align*}
\end{lemma}

With this Lemma, it is not hard to prove completeness:
\begin{proof}[Proof of completeness]
	Suppose $(C, V, \pH, \pUH, \pYL, \varepsilon) \in 
	\PiPubHide_Y$. 
	Fix $\alpha$ such that we have
	$\Pr_{y \leftarrow \Pd^C}[\Pd^C(y) \in (1 \pm 4\varepsilon) 
				\alpha 2^{-m}] \leq \sqrt{\varepsilon}$.
	By Claim~\ref{claim:smallMassAroundThreshold}, this holds 
	with probability at least $1-20\sqrt{\varepsilon
	}$ over the choice of $\alpha$.
	
	Since the prover is honest, it sends $u'$ and	$\cY'$ with
	correct witnesses. Then Lemma~\ref{lem:HighProbEstimates}
	implies that with probability at least $1-20\varepsilon$ 
	(a), (c),	(d), and (e) hold. Note that (b) always holds
	since the prover is honest.
	Finally, The lower bound protocol rejects with probability
	at most $\varepsilon/2$. 
\end{proof}

It remains to prove Lemma~\ref{lem:HighProbEstimates}. We 
defer the formal proof to Section~\ref{sec:PublicCoinHidingCompletenessDetails}, 
as the proof simply applies Chernoff and Hoeffding bounds.
Still, we give a short proof sketch.

\begin{proof}[Proof of Lemma~\ref{lem:HighProbEstimates} (Sketch)]
	Part (i) is a straightforward application of the Chernoff 
	bound. 
	
	Part (ii) also follows by the Chernoff bound, but
	here $|\cH'|/k$ may deviate from $\gUH$ by $O(\sqrt{\eps})$
	since this much mass may be close to the threshold, and be 
	cut off due to the rounding we introduce with the use of
	the labeling $u'$. 
	
	The proofs of (iii) and (iv) are applications of the
	Hoeffding bound, and again the $O(\sqrt{\eps})$ deviation
	comes in due to the rounding issues as described. 
	
	Finally, (v) is a straightforward application of the 
	Chernoff bound on (\ref{eqnp:5}).
\end{proof}

\subsection{Proof of Soundness: Overview}

\newcommand{\Loss}{\mathsf{Loss}}
\newcommand{\Gain}{\mathsf{Gain}}

Suppose 
$(C, V, \pH, \pUH, \pYL, \varepsilon) \in \PiPubHide_N$. 
Then the following lemma states that if there is not too much
probability mass around the threshold,
the verifier's checks (a)-(d) are true, 
the guarantees of the lower bound protocol hold,
and the high probability estimates for $u', \cY'$ and $\cL'$
hold, then the sum in the verifier's check (e) is close to 
$\gYL$. 
\begin{lemma}
	\label{lem:SoundnessGivenHighProbEst}
	Suppose 
	$(C, V, \pH, \pUH, \pYL, \varepsilon) \in \PiPubHide_N$. 
	Define the sets $\cS:=\{y_1, \ldots, y_k\}$, 
	$\cM:=\{y: \Pd^C(y) 
		\in (1\pm 4\varepsilon)\alpha 2^{-m}\}$, assume
	that the verifier's conditions (a)-(d) hold, and
	\begin{align}
			&\Pr_{y \leftarrow \Pd^C}[\Pd^C(y) \in (1 \pm 
																4\varepsilon) 
				\alpha 2^{-m}] \leq \sqrt{\varepsilon}
				\label{eqnp:8}\\
			&\forall i \in [k]: u'(i) = \infty \implies u(i) = \infty
				\label{eqnp:9p}\\
			&\forall i \in [k]: u(i)=\infty \lor |C^{-1}(y_i)| 
				> (1-\varepsilon/2)\cdot2^m\cdot 2^{-(u(i)+1)
						\varepsilon}
				\label{eqnp:9}\\	
			&\frac{|\cS \cap \cM|}{k} 
					< \frac{3\sqrt{\varepsilon}}{\alpha}
				\label{eqnp:10}\\
			&\frac{|\cY'|}{k} \in [\gUY \pm \varepsilon], 
			  \label{eqnp:11}\\
			&\frac{|\cH'|}{k} \in [\gUH \pm 
						3\sqrt{\varepsilon}], 
			  \label{eqnp:13}\\ 
			&\frac{1}{k} \sum_{i\in \cL'}
									2^m\cdot 2^{-u'(i)\varepsilon}\in [
									1-\gH \pm 5\sqrt{\varepsilon}],
				\label{eqnp:14}\\
		  &\frac{1}{k} \sum_{i \in \cL' \cap \cY'}
									2^m\cdot 2^{-u'(i)\varepsilon}\in [
									\gYL \pm 5\sqrt{\varepsilon}].
				\label{eqnp:15}
	\end{align}
	Then we have
	\begin{align*}
		\frac{1}{k}
					\sum_{i\in \cL\cap \cY}
					2^m 2^{-u(i)\varepsilon}
					\in [\gYL \pm 112\sqrt{\varepsilon}\alpha]		
	\end{align*}
\end{lemma}

This lemma allows us to prove soundness as follows.
\begin{proof}[Proof of soundness]
	Let $(C, V, \pH, \pUH, \pYL, \varepsilon) \in \PiPubHide_N$, 
	and fix $\alpha$ such that 
	$\Pr_{y \leftarrow \Pd}[\Pd(y) \in (1 \pm 4\varepsilon) 
				\alpha 2^{-m}] \leq \sqrt{\varepsilon}$.
	By Claim~\ref{claim:smallMassAroundThreshold}, this holds 
	with probability at least $1-20\sqrt{\varepsilon
	}$ over the choice of $\alpha$.
	
	We proceed to show that with probability at least
	$1-6\varepsilon$, all the assumptions of 
	Lemma~\ref{lem:SoundnessGivenHighProbEst} hold, or the 
	verifier rejects (with high probability).		
	We have that (\ref{eqnp:8}) holds by the above assumption. 
	Moreover, (\ref{eqnp:9p}) holds because $u'(i) = \infty$ 
	implies $|C^{-1}(y_i)|=0$ and thus the lower bound protocol
	rejects with probability $1$.	
	Furthermore, (\ref{eqnp:9}) holds with probability at least
	$1-\varepsilon/2$, by the soundness of the parallel lower 
	bound protocol. Finally, by 
	Lemma~\ref{lem:HighProbEstimates} and the union bound,
	we have that with probability at least $1-5\varepsilon$,
	all of	
	(\ref{eqnp:10}),
	(\ref{eqnp:11}),
	(\ref{eqnp:13}),
	(\ref{eqnp:14}),
	and (\ref{eqnp:15}) hold. Finally, either (a)-(d) hold, or
	the verifier rejects.
	
	If its assumptions hold, Lemma~\ref{lem:SoundnessGivenHighProbEst} gives that 
	$
		\frac{1}{k}
					\sum_{i\in \cL\cap \cY}
					2^m 2^{-u(i)\varepsilon}
					\in [\gYL \pm 112\sqrt{\varepsilon}\alpha]
	$.
	Together with the soundness assumption
	$\pYL \notin [\gYL \pm 117\sqrt{\varepsilon}\alpha]$,
	this implies 
	$
		\frac{1}{k}
					\sum_{i\in \cL\cap \cY}
					2^m 2^{-u(i)\varepsilon}
					\notin [\pYL \pm 5\sqrt{\varepsilon}\alpha]
	$,
	which gives that the verifier rejects in (e). 
	This shows that the verifier rejects with probability
	at least $1-6\varepsilon$. 
\end{proof}
	
It remains to prove Lemma~\ref{lem:SoundnessGivenHighProbEst}.
For this, we will use the notion of $\Loss$ and $\Gain$, which
is defined as follows:
\begin{definition}
	Given two mappings $u',u$ 
	from $[k]$ to $(m) \cup \{\infty\}$ and a set $\cA \subseteq
	[k]$,
	we define
	\begin{align*}
	\Loss_{\cA}(u',u) &:= \frac{1}{k} \sum_{i \in \cA: u'(i)<u(i)} 2^m (2^{-u'(i)\varepsilon}-2^{-u(i)\varepsilon}), \\
	\Gain_{\cA}(u',u) &:= \frac{1}{k} \sum_{i \in \cA: u'(i)>u(i)} 2^m (2^{-u(i)\varepsilon}-2^{-u'(i)\varepsilon}).
	\end{align*}
\end{definition}

Note that $\Gain$ and $\Loss$ are always positive. 
This notion is supposed to capture the change of probability
mass when using the labeling $u$ instead of the labeling 
$u'$, as described by the following claim. Its proof is not
hard, and we defer it to 
Section~\ref{sec:PublicCoinHidingSoundnessDetails}. 
\begin{claim}
	\label{claim:GainAndLoss}
	For any two mappings $u', u$ from $[k]$ to $(m) \cup 
	\{\infty\}$ and any $\cA \subseteq [k]$ we have
	$$
		\frac{1}{k}\sum_{i\in \cA}2^m 2^{-u(i)\varepsilon}
		=
		\frac{1}{k}\sum_{i\in \cA}2^m 2^{-u'(i)\varepsilon}
		+ \Gain_{\cA}(u',u) - \Loss_{\cA}(u',u).
	$$
\end{claim}

We establish the following sequence of intermediate claims which will then allow us to prove 
Lemma~\ref{lem:SoundnessGivenHighProbEst}.
\begin{claim}
	\label{claim:SoundnessClaimsGivenHighProbEst}
	Under the conditions of 
	Lemma~\ref{lem:SoundnessGivenHighProbEst}, we have:
	\begin{align*}
	\text{(i)} \qquad 
			&\forall i\in [k]: u(i) \geq u'(i) - 1 \\
		\text{(ii)} \qquad 
			&|\cL' \setminus \cL| 
					\leq 3k\sqrt{\varepsilon}\\
		\text{(iii)} \qquad &|\cL \setminus \cL'| 
					\leq 19k\sqrt{\varepsilon} \\
		\text{(iv)} \qquad 
						&\frac{1}{k}\sum_{i\in \cL' \setminus \cL}
															2^m2^{-u'(i)\varepsilon}
						\leq 6\sqrt{\varepsilon}\alpha
						\\
						&\frac{1}{k}\sum_{i\in \cL \setminus \cL'}
															2^m2^{-u(i)\varepsilon}
						\leq 38\sqrt{\varepsilon}\alpha \\
		\text{(v)} 
					\qquad &
					\frac{1}{k}\sum_{i\in \cL \cap \cL'}
					 2^m2^{-u(i)\varepsilon}
					\in [1-\gH \pm 44\sqrt{\varepsilon}\alpha] \\
					&
					\frac{1}{k}\sum_{i\in \cL \cap \cL'}
					 2^m2^{-u'(i)\varepsilon}
					\in [1-\gH \pm 11\sqrt{\varepsilon}\alpha] \\
		\text{(vi)} 
					\qquad &
					\Gain_{\cL \cap \cL'}(u',u) 
						\leq 4\varepsilon
					\\
					&\Loss_{\cL \cap \cL'}(u',u) 
						\leq 59\sqrt{\varepsilon}\alpha 
	\end{align*}
\end{claim}
We defer the proof to 
Section~\ref{sec:PublicCoinHidingSoundnessDetails}, and only
sketch the proof here.
\begin{proof}[Proof (Sketch)]
	Part (i) holds because the claimed probabilities as given
	by $u'$ cannot be too large, as the lower bound guarantees
	(\ref{eqnp:9p}) and (\ref{eqnp:9}) hold. 
	
	Part (ii) holds because in order to claim that some light 
	$y_i, i \in \cL'$ is heavy (i.e.~$i \notin \cL$), its
	probability must be close to the threshold, which holds only
	for an $\Theta(\sqrt{\eps})$-fraction of the $y_i$'s. This
	holds because the lower bounds are accurate up to a factor
	of even $(1-\eps/2)$. 
	
	Then (iii) follows from (ii), as by condition (c) and 
	$\pUH \in [\gUH \pm \Theta(\sqrt(\eps))]$ we have that
	$|\cL|$ and $|\cL'|$ can differ by at most 
	$k\cdot \Theta(\sqrt{\eps})$. 
	
	Part (iv) is then a direct consequence of (ii) and (iii),
	using the fact that $i \in \cL'$ and $i\in \cL$,
	respectively. 
	
	To see (v), we note that the sum over $\cL'$ using $u'$ is
	close to $1-\gH$ by (\ref{eqnp:14}), and the sum over $\cL$ 
	using $u$ is close to $1-\gH$ by the guarantee $\pH \in 
	[\gH \pm 4/5\sqrt{\eps}]$ and (d). Applying (ii) and (iii) gives
	the result.
	
	Finally, to prove (vi) we note that the two sums in 
	(v) are close, which implies they have small difference. 
	By definition, this difference is exactly 
	$\Gain_{\cL\cap\cL'}(u',u) -
	\Loss_{\cL\cap\cL'}(u',u)$. Since (i) allows to upper bound
	the $\Gain$, we get the claim.	
\end{proof}

Finally, this allows us to prove our goal as follows. We only sketch 
the proof and defer the details to Section~\ref{sec:PublicCoinHidingSoundnessDetails}. 
\begin{proof}[Proof of 
	Lemma~\ref{lem:SoundnessGivenHighProbEst} (Sketch)]
	We will directly refer to (i)-(vi) as given by 
	Claim~\ref{claim:SoundnessClaimsGivenHighProbEst}. 
	We make the following observations:
	\begin{enumerate}[(1)]
		\item	We may as well consider the sum over 
			$\cY'$ instead of $\cY$: this only induces an error 
			of order $O(\eps \alpha)$, because
			the prover must provide witnesses (see (a) and (b)), and
			the set $|\cY|$ must still be big (\ref{eqnp:11}). 
	
	  \item We may as well consider the sum over
			$\cL \cap \cL'$ instead of $\cL$ or $\cL'$: this is a
			direct consequence of (iv), and induces an error of at
			most $O(\sqrt{\eps}\alpha)$. 
	
		\item By definition, 
			$\Gain_{\cL\cap \cL' \cap \cY'}(u',u) 
			\leq \Gain_{\cL\cap \cL'}(u',u),
			\Loss_{\cL\cap \cL' \cap \cY'}(u',u) 
			\leq \Loss_{\cL\cap \cL'}(u',u)$, 
			and thus both are bounded by
			$\Theta(\sqrt{\eps}\alpha)$ by (vi). 
	\end{enumerate}
	
	This allows us to conclude (we put the actual constants 
	to be explicit)
	\begin{align*}
		&\frac{1}{k}\sum_{i\in \cL \cap \cY}2^m2^{-u(i)\varepsilon}
		\stackrel{\text{(2)}}{\in} 
		[\frac{1}{k}\sum_{i\in \cL \cap\cL'\cap \cY}
					2^m2^{-u(i)\varepsilon} \pm 38\sqrt{\varepsilon}
					\alpha] \\
		&\stackrel{\text{(1)}}{\subseteq} 
		[\frac{1}{k}\sum_{i\in \cL \cap\cL'\cap \cY'}
					2^m2^{-u(i)\varepsilon} \pm 42\sqrt{\varepsilon}
					\alpha]
		\subseteq
		[\frac{1}{k}\sum_{i\in \cL \cap\cL'\cap \cY'}
					2^m2^{-u'(i)\varepsilon} \pm 101\sqrt{\varepsilon}
					\alpha]\\
		&\stackrel{\text{(2)}}{\subseteq} 
		[\frac{1}{k}\sum_{i\in \cL'\cap \cY'}
					2^m2^{-u'(i)\varepsilon} \pm 107\sqrt{\varepsilon}
					\alpha]
		\stackrel{\text{(\ref{eqnp:15})}}{\subseteq} 
		[\gYL \pm 112\sqrt{\varepsilon}
					\alpha],
	\end{align*}
	where the third step follows by (3) and the 
	definition of $\Gain$ and $\Loss$. 
\end{proof}

\subsection{Proof of Completeness: the Details}
\label{sec:PublicCoinHidingCompletenessDetails}

It is not hard to see that only few $y$'s in the support
have $\Pd^C(y)$ close to the threshold:
\begin{claim}
	\label{claim:FewUniformInThreshold}
	Suppose 
	\begin{align}
	  \Pr_{y \leftarrow \Pd^C}[\Pd^C(y) \in (1 \pm 4\varepsilon) 
				\alpha 2^{-m}] \leq \sqrt{\varepsilon}.
				\label{eqnp:3}
	\end{align}
	Then
	\begin{align*}
		\text{(i)}\qquad &
				|\{y: \Pd^C(y) \in (1\pm 4\varepsilon)\alpha 2^{-m}\}|
				\leq \frac{2\sqrt{\varepsilon} \cdot 2^m}{\alpha} \\
		\text{(ii)}\qquad &
				|\{y: 2^{-u'(y)\varepsilon} \geq \alpha 2^{-m}
								> 2^{-(u'(y)+1)\varepsilon}\}| \leq 
						 \frac{2\sqrt{\varepsilon}\cdot 2^m}{\alpha}
	\end{align*}
\end{claim}

\begin{proof}[Proof of Claim~\ref{claim:FewUniformInThreshold}]
	We first prove (i). Let $\cM_1$ be the set in (i). 
	Now (\ref{eqnp:3}) gives that 
	$|\cM_1| \leq \frac{\sqrt{\varepsilon}}
									 {(1-4\varepsilon)\alpha 2^{-m}} 
	 \leq \frac{2\sqrt{\varepsilon}\cdot 2^m}{\alpha}$ 
	(for $\varepsilon \leq 1/8$).	
	
	To see (ii), let $\cM_2$ be the set in (ii). 
	Using Claim~\ref{claim:cDandUIntervals}, we find that
	for any $y \in \cM_2$ we have
	$$ \Pd^C(y) \in [2^{-\varepsilon}\alpha 2^{-m},
									 2^{\varepsilon} \alpha 2^{-m}]
							\subseteq (1\pm 2\varepsilon) \alpha 2^{-m},$$
	and thus $\cM_2 \subseteq \cM_1$, which proves the claim.
\end{proof}

The following claim describes how we apply the Hoeffding bound,
we will use it to prove parts (iii) and (iv) of 
Lemma~\ref{lem:HighProbEstimates}.
\begin{claim}
	\label{claim:AppliedHoeffding}
	Consider the set $\cAl := \{y:  2^{-(u'(y)+1)\varepsilon} <
	\alpha 2^{-m} \}$, let 
	$\cA$ be any subset of $\{0,1\}^m$, and define
	\begin{align*}
		M:=\sum_{y\in \cAl \cap \cA}
					2^{-u'(y)\varepsilon}, \qquad \qquad
		X_i :=
			\begin{cases}
				2^m\cdot 2^{-u'(y_i)\varepsilon}
						& \text{if } y_i \in \cAl \cap \cA,\\
				0 & \text{otherwise.} 
			\end{cases}
	\end{align*}
	Then we have for any $\delta > 0$ that
	$$ 
		\Pr_{y_1, \ldots, y_k \in \{0,1\}^{km}}\left[
			\frac{1}{k}\sum_{i\in [k]}X_i
			\notin [M \pm \delta]
		\right]
		\leq
		2\cdot \exp\left(
			-\frac{2k\delta^2}{(1+2\varepsilon)^2\alpha^2}
			\right).
	$$ 
\end{claim}
\begin{proof}
	For any $i \in [k]$ we find 
	$\Exp_{y_i \in \{0,1\}^m}[X_i] = M$, and since
	$2^{-(u'(y_i)+1)\varepsilon} 
												 < \alpha 2^{-m}$
	implies that 
	$
		2^m\cdot 2^{-u'(y_i)\varepsilon} 
		= 2^m \cdot 2^{\varepsilon}\cdot 2^{-(u'(y_i)+1)\varepsilon}
		\leq (1+2\varepsilon)\alpha
	$,
	we have $X_i \in [0,(1+2\varepsilon)\alpha]$. 
	As 
	$$
		\Exp_{y_1, \ldots, y_k}[\frac{1}{k}\sum_{i\in [k]}X_i] 
		= \Exp_{y_1}[X_1] = M,
	$$
	the Hoeffding bound (Lemma~\ref{lem:HoeffdingBound}) gives the claim.
\end{proof}	

\begin{proof}[Proof of Lemma~\ref{lem:HighProbEstimates} (i)]
	Let $Y_i$ be the indicator variable for the
	event $y_i \in V$. Then $\Pr[Y_i =1] = \gUY$, 
	we have $|\cY'| = \sum_{i\in [k]}Y_i$, and so the Chernoff 
	bound (Lemma~\ref{lem:ChernoffBound}) gives	
	\begin{align*}
		\Pr_{y_1, \ldots, y_k}
			\left[\frac{|\cY'|}{k} \notin [\gUY \pm \varepsilon]
						\right]
		 =
		 \Pr\left
		[\sum_{i\in [k]} Y_i \notin [\gUY\pm\varepsilon]k\right]
			\leq 2\exp(-\frac{\varepsilon^2k}{2}) \leq \varepsilon.
	\end{align*}
\end{proof}
				
\begin{proof}[Proof of Proof of Lemma~\ref{lem:HighProbEstimates} (ii)]
	Let $X_i$ be the indicator variable for the
	event $y_i \in \{y: 2^{-(u(y)+1)\varepsilon} \geq 
	\alpha 2^{-m}\}$. We first show that $p:=\Pr[X_i = 1]$ is
	close to $\gUH$. 
	Using Claim~\ref{claim:cDandUIntervals}, we get
	\begin{align*}
		p &= 
		\frac{|\{y: 2^{-(u(y)+1)\varepsilon} \geq \alpha 2^{-m
			 }\}| }{2^m}
		\leq		
		\frac{|\{y: \Pd^C(y) \geq \alpha 2^{-m}\}|}{2^m} = \gUH,\\
		\gUH &\leq 
		\frac{|\{y: 2^{-u(y)\varepsilon} \geq \alpha 2^{-m
			 }\}| }{2^m}
		  \\
		& = \frac{|\{y: 2^{-(u(y)+1)\varepsilon} \geq \alpha 2^{-m
			 }\}| }{2^m} +
			\frac{|\{y: 2^{-u(y)\varepsilon} \geq \alpha 2^{-m
			 } > 2^{-(u(y)+1)\varepsilon}\}|}{2^m} \\
		&\leq
		 p + 2\sqrt{\varepsilon},
	\end{align*}
	where we applied Claim~\ref{claim:FewUniformInThreshold} to 
	obtain the last inequality.
	This shows that 
	\begin{align}
		\label{eqnp:4}
		p \in [\gUH-2\sqrt{\varepsilon}, \gUH].
	\end{align}	
	
	Now
	we have $|\cH'| = \sum_{i\in [k]}X_i$, and so the 
	Chernoff bound (Lemma~\ref{lem:ChernoffBound}) gives
	\begin{align}
		\Pr_{y_1, \ldots, y_k}
			\left[\frac{|\cH'|}{k} \notin [p \pm \varepsilon]\right]
		 =
		 \Pr\left[\sum_{i\in [k]} X_i \notin [p \pm \varepsilon] 
		k\right]
			\leq 2\cdot \exp(-\frac{\varepsilon^2k}{2}) 
			\leq \varepsilon.
	\end{align}
	Plugging (\ref{eqnp:4}) into the above gives the claim.
\end{proof}
	
\begin{proof}[Proof of Proof of Lemma~\ref{lem:HighProbEstimates} (iii)]
	Let $y$ denote a bitstring in $\{0,1\}^m$.
	First note that
	\begin{align}
		1-\gH 
		&= 1- \sum_{y: \Pd^C(y) \geq \alpha 2^{-m}}
					\Pd^C(y) 
		= \sum_{y: \Pd^C(y) < \alpha 2^{-m}}
					\Pd^C(y) \\
		&= \sum_{y: \Pd^C(y) < 2^{\varepsilon} \alpha 2^{-m}}
					\Pd^C(y)
			-
			\underbrace{\sum_{y: \Pd^C(y) \in [\alpha 2^{-m},
														 2^{\varepsilon} \alpha 2^{-m})} \Pd^C(y)}_{\in [0, \sqrt{\varepsilon}]} \label{eqnp:6}
\end{align}
The above sum is indeed in $[0,\sqrt{\varepsilon}]$, as
	$[\alpha 2^{-m}, 2^{\varepsilon} \alpha2^{-m}) \subseteq
	 (1\pm 4\varepsilon) \alpha 2^{-m}$, 
and thus we can use assumption (\ref{eqnp:5}).
Using Claim~\ref{claim:cDandUIntervals}, we find
\begin{align*}
		&\sum_{y: 2^{-(u'(y)+1)\varepsilon} < \alpha 2^{-m}}
					\Pd^C(y) 
		\geq
		\sum_{y: \Pd^C(y) < \alpha 2^{-m}}
					\Pd^C(y) 
		=1-\gH \\
		&\qquad \stackrel{\text{(\ref{eqnp:6})}}{\geq}
		\sum_{y: \Pd^C(y) < 2^{\varepsilon} \alpha 2^{-m}}
					\Pd^C(y)
		- \sqrt{\varepsilon}
		\geq 
		\sum_{y: 2^{-u'(y)\varepsilon} 
						 < 2^{\varepsilon} \alpha 2^{-m}}
					\Pd^C(y)
		- \sqrt{\varepsilon}\\
		&\qquad =
		\sum_{y: 2^{-(u'(y)+1)\varepsilon} 
						 < \alpha 2^{-m}}
					\Pd^C(y)
		- \sqrt{\varepsilon}.	
\end{align*}
This implies that $1-\gH$ is in the interval
\begin{align*}
 \left[
		\sum_{y: 2^{-(u'(y)+1)\varepsilon} < \alpha 2^{-m}}
		\Pd^C(y) \pm \sqrt{\varepsilon} \right]
 \subseteq \Bigl[
		2^{\pm \varepsilon}
		\underbrace{\sum_{y: 2^{-(u'(y)+1)\varepsilon} < \alpha 2^{-m}}
		2^{-u'(y)\varepsilon}}_{=:M}\pm \sqrt{\varepsilon}\Bigr], 
\end{align*}
where the last step above follows by Claim~\ref{claim:cDandUIntervals}. From this we get
\begin{align*}
	M &\leq 2^{\varepsilon}(1-\gH)
				 +2^{\varepsilon}\sqrt{\varepsilon}
		\leq (1+2\varepsilon)(1-\gH)+ 2\sqrt{\varepsilon}\\
		&\leq (1-\gH) + 2\varepsilon + 2\sqrt{\varepsilon} 
		\leq (1-\gH) + 4\sqrt{\varepsilon}, \\
	M &\geq 2^{-\varepsilon}(1-\gH) 
				 -2^{-\varepsilon}\sqrt{\varepsilon}
		\geq (1-2\varepsilon)(1-\gH) - \sqrt{\varepsilon}\\
		&\geq (1-\gH)- 2\varepsilon - \sqrt{\varepsilon}
	  \geq (1-\gH)- 3\sqrt{\varepsilon},
\end{align*}
and thus 
\begin{align}
	M\in [(1-\gH)\pm 4\sqrt{\varepsilon}].
	\label{eqnp:49}
\end{align}

	Applying Claim~\ref{claim:AppliedHoeffding}
	to $M$ as defined above for $\cA = \{0,1\}^m$, we get that
	$$
		\Pr_{y_1, \ldots, y_k}
		\left[
			\frac{1}{k} \sum_{i \in \cL'}2^m 2^{-u'(i)\varepsilon}
			\notin
			\left[
				M	\pm \varepsilon
			\right]
		\right] \leq 2\exp\left(-\frac{2k\varepsilon^2}{(1+2\varepsilon)^2\alpha^2}\right) \leq \varepsilon.
	$$
By (\ref{eqnp:49}), we have 
$[M\pm \varepsilon] \subseteq [1-\gH \pm 5\sqrt{\varepsilon}]$,
which gives the claim.
\end{proof}

\begin{proof}[Proof of Proof of Lemma~\ref{lem:HighProbEstimates} (iv)]
The proof is analogous to the proof of (iii): we find
	\begin{align}
		\gYL &= \sum_{y: \Pd^C(y) < \alpha 2^{-m} \land y \in V} 
				\Pd^C(y) 
			\in \Bigl[
	 2^{\pm \varepsilon}
			\sum_{y: 2^{-(u(y)+1)\varepsilon} 
					  < \alpha 2^{-m} \land y \in V}
					2^{-u(y)\varepsilon}\pm \sqrt{\varepsilon}\Bigr],
			\label{eqnp:2}
	\end{align}	
	where the $\sqrt{\varepsilon}$ deviation can be seen 
	as in (iii), since the sum here has only less summands.	
	Thus, applying Claim~\ref{claim:AppliedHoeffding}
	to the sum in (\ref{eqnp:2}) for 
	$\cA = \{y: y\in V\}$ gives the claim.
\end{proof}
\begin{proof}[Proof of Lemma~\ref{lem:HighProbEstimates} (v)]
	Claim~\ref{claim:FewUniformInThreshold} gives that
	$|\cM| \leq \frac{2\sqrt{\varepsilon}\cdot 2^m}{\alpha}$.
	Let $X_i$ be the indicator random variable for the event
	$y_i \in \cM$. Then $|\cS \cap \cM| 
	= \sum_{i\in [k]}X_i$, and $p:=\Exp_{y_i}[X_i] \leq 
	\frac{2\sqrt{\varepsilon}}{\alpha}$. The Chernoff	bound (Lemma~\ref{lem:ChernoffBound}) gives
	\begin{align*}
		\Pr_{\cS}
				\left[
					\frac{|\cS \cap \cM|}{k} \geq \frac{3\sqrt{\varepsilon}}{\alpha}
				\right]
		&\leq 
			\Pr_{\cS}
				\left[
					\frac{|\cS \cap \cM|}{k} \geq p + \frac{
					\sqrt{\varepsilon}}{\alpha}
				\right]\\
		&= \Pr_{\cS}
				\left[
					\sum_{i\in [k]}X_i \geq (p + \frac{
					\sqrt{\varepsilon}}{\alpha})k
				\right]
		\leq \exp\left(-\frac{\varepsilon k}{2 \alpha^2}\right)
		\leq \varepsilon.\qedhere
	\end{align*}
\end{proof}

\subsection{Proof of Soundness: the Details}
\label{sec:PublicCoinHidingSoundnessDetails}

\begin{proof}[Proof of Claim~\ref{claim:GainAndLoss}]
	By definition of $\Gain$ and $\Loss$, the right hand side
	equals
	\begin{align*}
		&\frac{1}{k}\sum_{i\in \cA}2^m 2^{-u'(i)\varepsilon}
		+\frac{1}{k} \sum_{i \in \cA: u'(i)>u(i)} 2^m (2^{-u(i)\varepsilon}-2^{-u'(i)\varepsilon})\\
		&\qquad\qquad\qquad-\frac{1}{k} \sum_{i \in \cA: u'(i)<u(i)} 2^m (2^{-u'(i)\varepsilon}-2^{-u(i)\varepsilon}) \\
		&\quad 
			= \frac{2^m}{k}\Bigl(
				\sum_{i\in \cA} 2^{-u'(i)\varepsilon}
				+ \sum_{i \in \cA: u'(i)>u(i)} 2^{-u(i)\varepsilon}
				- \sum_{i \in \cA: u'(i)>u(i)} 2^{-u'(i)\varepsilon}\\
				&\qquad \qquad \qquad \qquad \quad \,\,
				- \sum_{i \in \cA: u'(i)<u(i)} 2^{-u'(i)\varepsilon}
				+ \sum_{i \in \cA: u'(i)<u(i)} 2^{-u(i)\varepsilon} 
			\Bigr) \\
		&\quad 
			= \frac{2^m}{k}\Bigl(
			\sum_{i\in \cA} 2^{-u'(i)\varepsilon}
				+ \sum_{i \in \cA: u'(i) \neq u(i)} 
					2^{-u(i)\varepsilon}
				- \sum_{i \in \cA: u'(i) \neq u(i)} 
					2^{-u'(i)\varepsilon}
			\Bigr) \\
		&\quad
		  =	\frac{1}{k}\sum_{i\in \cA}2^m 2^{-u(i)\varepsilon}.
			\qedhere
	\end{align*}
\end{proof}

\begin{proof}[Proof of Claim~\ref{claim:SoundnessClaimsGivenHighProbEst} (i)]
	%If $u'(i)=\infty$, then $|C^{-1}(y_i)| = 0$ and thus
	%the verifier rejects (with probability $1$) 
	%in the lower bound protocol if $u(i) \neq \infty$. 
	In case $u'(i) = \infty$, the claim follows by 
	(\ref{eqnp:9p}). So suppose $u'(i) < \infty$. 
	Towards a contradiction assume $u(i) \leq u'(i)-2$.	
	Then
	\begin{align*}
		|C^{-1}(y_i)| &= 2^m \Pd^C(y_i) 
			\leq 2^m\cdot 2^{-u'(i)\varepsilon}
			\leq 2^m\cdot 2^{-(u(i)+2)\varepsilon}\\
			&=2^{-\varepsilon}\cdot 2^m\cdot 2^{-(u(i)+1)\varepsilon
			}
			\leq (1-\varepsilon/2)\cdot 2^m\cdot 
					2^{-(u(i)+1)\varepsilon},
	\end{align*}
	where the first inequality follows by 
	Claim~\ref{claim:cDandUIntervals}. As $u(i) < \infty$ by
	assumption, this contradicts (\ref{eqnp:9}).
\end{proof}

\begin{proof}[Proof of Claim~\ref{claim:SoundnessClaimsGivenHighProbEst} (ii)]
	By (\ref{eqnp:10}), we have 
	$|(\cL' \setminus \cL) \cap \cM| \leq \frac{3k\sqrt{\varepsilon}}{\alpha}$. 
	We show that $(\cL' \setminus \cL) \cap \overline{\cM} =
	\emptyset$. Towards a contradiction assume there is some
	 $y_i \in (\cL' \setminus \cL) \cap \overline{\cM}$. We have
	\begin{align}
		&y_i \notin \cM \implies \Pd^C(y) \notin 
				(1\pm 4\varepsilon)\alpha 2^{-m},
				\label{eqnp:47}\\
		&y_i \in \cL' \implies 
				2^{-(u'(i)+1)\varepsilon} < \alpha 2^{-m},
				\label{eqnp:48} \\
		&y_i \notin \cL \implies 2^{-(u(i)+1)\varepsilon} 
				\geq \alpha 2^{-m},
				\label{eqnp:18} \\
		&\Pd^C(y_i) 
			\stackrel{\text{Claim~\ref{claim:cDandUIntervals}}}
				{\leq} 2^{-u'(i)\varepsilon} 
			= 2^{\varepsilon}2^{-(u'(i)+1)\varepsilon}
			\stackrel{\text{(\ref{eqnp:48})}}{\leq} 
				(1+2\varepsilon)\alpha 2^{-m}.
			\label{eqnp:16}
	\end{align}
	Now (\ref{eqnp:47}) and (\ref{eqnp:16}) give that
	\begin{align}
		\Pd^C(y_i) \leq (1-4\varepsilon)\alpha 2^{-m}.
		\label{eqnp:17}
	\end{align}
	Then we find
	\begin{align*}
		2^{-(u'(i) +1)\varepsilon} 
		&\stackrel{\text{Claim~\ref{claim:cDandUIntervals}}}{\leq} 
		\Pd^C(y_i) 
		\stackrel{\text{(\ref{eqnp:17})}}{\leq} 
		(1-4\varepsilon)\alpha 2^{-m}
		\stackrel{\text{(\ref{eqnp:18})}}{\leq} 
		(1-4\varepsilon)2^{-(u(i)+1)\varepsilon} \\
		&\leq 2^{-2\varepsilon} 2^{-(u(i)+1)\varepsilon} 
		= 2^{-(u(i)+3)\varepsilon},
	\end{align*}	
	which implies $u(i) < u'(i)-1$, and thus contradicts
	(i). 
\end{proof}

\begin{proof}[Proof of Claim~\ref{claim:SoundnessClaimsGivenHighProbEst} (iii)]	
	(\ref{eqnp:13}), condition (c), and 
	$\pUH \in [\gUH\pm 10\sqrt{\varepsilon}]$ give that
	\begin{align}
		|\cL'|/k &\in [1-\gUH \pm 3\sqrt{\varepsilon}], 
			\label{eqnp:19}\\
		|\cL|/k  &\in [1-\pUH  \pm 3\sqrt{\varepsilon}]
						\subseteq [1-\gUH  \pm 13\sqrt{\varepsilon}].
		  \label{eqnp:20}
	\end{align}	
	Now we find
	\begin{align}
		|\cL'\cap \cL| 
		= |\cL'| - |\cL'\setminus \cL|
		\stackrel{\text{(ii)}}{\geq} 
		|\cL'| - k\cdot 3\sqrt{\varepsilon}
		\stackrel{\text{(\ref{eqnp:19})}}{\geq} 
		k\cdot(1-\gUH-6 \sqrt{\varepsilon}),
		\label{eqnp:21}
	\end{align}
	and thus
	\begin{align*}
		|\cL \setminus \cL'| 
		&= |\cL| - |\cL \cap \cL'|
		\stackrel{\text{(\ref{eqnp:21})}}{\leq} 
		  |\cL| - k\cdot(1-\gUH-6 \sqrt{\varepsilon}) \\
		&\stackrel{\text{(\ref{eqnp:20})}}{\leq} 
		  k\cdot(1-\gUH + 13\sqrt{\varepsilon}) 
			 - k\cdot(1-\gUH-6 \sqrt{\varepsilon})
		= k\cdot 19\sqrt{\varepsilon}. \qedhere
	\end{align*}
\end{proof}

\begin{proof}[Proof of Claim~\ref{claim:SoundnessClaimsGivenHighProbEst} (iv)]
	By definition of $\cL$ and $\cL'$ we find
	\begin{align}
		\frac{1}{k}\sum_{i\in \cL' \setminus \cL}
										  2^m2^{-u'(i)\varepsilon}
		&\leq 
		\frac{1}{k}|\cL'\setminus\cL| 2^m 2^{\varepsilon}
							\alpha 2^{-m}
		\stackrel{\text{(ii)}}{\leq}
		\frac{1}{k}k\cdot 2^{\varepsilon} 
				\cdot 3\sqrt{\varepsilon}\alpha		\nonumber \\
		&=2^{\varepsilon}\cdot 3\sqrt{\varepsilon}\alpha
		\leq 6\sqrt{\varepsilon}\alpha,
		\label{eqnp:22}
		\\
		\frac{1}{k}\sum_{i\in \cL \setminus \cL'}
										  2^m2^{-u(i)\varepsilon}
		&\leq 
		\frac{1}{k}|\cL\setminus\cL'|2^m 
				2^{\varepsilon} \alpha 2^{-m}
		\stackrel{\text{(iii)}}{\leq}
		\frac{1}{k}k\cdot 2^{\varepsilon} 
				\cdot 19\sqrt{\varepsilon}\alpha	\nonumber	\\
		&=2^{\varepsilon}\cdot 19\sqrt{\varepsilon}\alpha
		\leq 38 \sqrt{\varepsilon}\alpha.
		\label{eqnp:23}
	\end{align}
\end{proof}

\begin{proof}[Proof of Claim~\ref{claim:SoundnessClaimsGivenHighProbEst} (v)]
	We find
	\begin{align*}
		\frac{1}{k}\sum_{i\in \cL}2^m2^{-u(i)\varepsilon}
		&= \frac{1}{k}\sum_{i\in \cL \cap \cL'}
										  2^m2^{-u(i)\varepsilon}
		  +
			\frac{1}{k}\sum_{i\in \cL \setminus \cL'}
										  2^m2^{-u(i)\varepsilon}\\
	  &\stackrel{\text{(\ref{eqnp:23})}}{\in} 
				\frac{1}{k}\sum_{i\in \cL \cap \cL'}
												2^m2^{-u(i)\varepsilon}
				+ [0, 38\sqrt{\varepsilon}\alpha],
		\\
		\frac{1}{k}\sum_{i\in \cL'}2^m2^{-u'(i)\varepsilon}
		&= \frac{1}{k}\sum_{i\in \cL \cap \cL'}
										  2^m2^{-u'(i)\varepsilon}
		  +
			\frac{1}{k}\sum_{i\in \cL' \setminus \cL}
										  2^m2^{-u'(i)\varepsilon}\\
	  &\stackrel{\text{(\ref{eqnp:22})}}{\in} 
				\frac{1}{k}\sum_{i\in \cL \cap \cL'}
												2^m2^{-u'(i)\varepsilon}
				+ [0, 6\sqrt{\varepsilon}\alpha].
	\end{align*}
	Applying assumptions (d) and $\pH \in [\gH \pm \frac{4}{5}
	\sqrt{\varepsilon}]$
	to the first inclusion 
	and assumption (\ref{eqnp:14}) to the second
	inclusion	gives the claim.
\end{proof}

\begin{proof}[Proof of Claim~\ref{claim:SoundnessClaimsGivenHighProbEst} (vi)]
	Claim~\ref{claim:GainAndLoss} gives 
	\begin{align}
		\frac{1}{k}\sum_{i\in \cL\cap\cL'}2^m 2^{-u(i)\varepsilon}
		=
		\frac{1}{k}\sum_{i\in \cL\cap\cL'}2^m 2^{-u'(i)\varepsilon}
		+ \Gain_{\cL\cap\cL'}(u',u) - \Loss_{\cL\cap\cL'}(u',u).
		\label{eqnp:24}
	\end{align}
	Furthermore, (v) gives
	\begin{align}
		\frac{1}{k}\sum_{i\in \cL \cap \cL'}
					 2^m2^{-u(i)\varepsilon}
					\in 
					\left[\frac{1}{k}\sum_{i\in \cL \cap \cL'}
					 2^m2^{-u'(i)\varepsilon}\pm 55\sqrt{\varepsilon}
					\alpha\right]
					\label{eqnp:25}
	\end{align}
	Combining (\ref{eqnp:24}) and (\ref{eqnp:25}) then gives
	\begin{align}
		\Gain_{\cL\cap\cL'}(u',u) - \Loss_{\cL\cap\cL'}(u',u)
		\in [0\pm 55\sqrt{\varepsilon}\alpha]
		\label{eqnp:26}
	\end{align}
	Now we find
	\begin{align}
		\Gain_{\cL\cap\cL'}(u',u)
		&= \frac{1}{k} \sum_{i \in \cL\cap\cL': u'(i)>u(i)} 
				2^m (2^{-u(i)\varepsilon}-2^{-u'(i)\varepsilon}) 
				\nonumber \\
		&\stackrel{\text{(i)}}{\leq}
			\frac{1}{k} \sum_{i \in \cL\cap\cL': u'(i)>u(i)} 
				2^m (2^{-(u'(i)-1)\varepsilon}-2^{-u'(i)\varepsilon})
				\nonumber \\
		&= \frac{1}{k} \sum_{i \in \cL\cap\cL': u'(i)>u(i)} 
				2^m (2^{\varepsilon}\cdot 2^{-u'(i)\varepsilon}-
														2^{-u'(i)\varepsilon})
														\nonumber \\
		&= \frac{1}{k} \sum_{i \in \cL\cap\cL': u'(i)>u(i)} 
				2^m (2^{\varepsilon}-1) 2^{-u'(i)\varepsilon}
				\nonumber \\
		&\leq
		  2\varepsilon 
			\frac{1}{k} \sum_{i \in \cL\cap\cL': u'(i)>u(i)} 
				2^m 2^{-u'(i)\varepsilon}\nonumber \\
	  &\stackrel{\text{(\ref{eqnp:14})}}{\leq} 
			2\varepsilon(1-\gH +5\sqrt{\varepsilon})
		\leq 4 \varepsilon.
		\label{eqnp:27}
	\end{align}
	Now we get
	\begin{align*}
		\Loss_{\cL\cap\cL'}(u',u) 
		\stackrel{\text{(\ref{eqnp:26})}}{\leq }
		\Gain_{\cL\cap\cL'}(u',u) + 55\sqrt{\varepsilon}\alpha
		\stackrel{\text{(\ref{eqnp:27})}}{\leq }
		59\sqrt{\varepsilon}\alpha. 
	\end{align*}
\end{proof}

\begin{proof}[Proof of 
	Lemma~\ref{lem:SoundnessGivenHighProbEst}]
	Throughout the proof we will write (i)-(vi) to refer
	to the corresponding items of Claim~\ref{claim:SoundnessClaimsGivenHighProbEst}.
	Assumptions (a), (b) and (\ref{eqnp:11}), 
	imply that
	\begin{align}
		&\cY \subseteq \cY' \label{eqnp:28}\\
		&|\cY' \setminus \cY| \leq 2\varepsilon k,\label{eqnp:29}
	\end{align}
	and so we find
	\begin{align*}
		\frac{1}{k}\sum_{i\in \cL \cap \cL'\cap \cY'}
					 2^m2^{-u(i)\varepsilon}
		&\stackrel{\text{(\ref{eqnp:28})}}{\geq}	
		\frac{1}{k}\sum_{i\in \cL \cap \cL'\cap \cY}
					 2^m2^{-u(i)\varepsilon} \\
		&= \frac{1}{k}\sum_{i\in \cL \cap \cL'\cap \cY'}
					 2^m2^{-u(i)\varepsilon}
		  -
			\frac{1}{k}\sum_{i\in \cL \cap \cL'\cap (\cY'
			\setminus \cY)}
					 2^m2^{-u(i)\varepsilon} \\
		&\stackrel{(i\in \cL)}{\geq}	
			\frac{1}{k}\sum_{i\in \cL \cap \cL'\cap \cY'}
					 2^m2^{-u(i)\varepsilon}
		  -\frac{1}{k} |\cY' \setminus \cY| 2^m 
				2^{\varepsilon}\alpha 2^{-m}\\
		&\stackrel{\text{(\ref{eqnp:29})}}{\geq}	
			\frac{1}{k}\sum_{i\in \cL \cap \cL'\cap \cY'}
					 2^m 2^{-u(i)\varepsilon}
		  - 4\varepsilon \alpha.
	\end{align*}
	This gives
	\begin{align}
		\frac{1}{k}\sum_{i\in \cL \cap \cL'\cap \cY}
					 2^m2^{-u(i)\varepsilon}
		\in [\frac{1}{k}\sum_{i\in \cL \cap \cL'\cap \cY'}
					 2^m 2^{-u(i)\varepsilon} \pm 4\varepsilon \alpha].
		\label{eqnp:30}
	\end{align}
	We find
	\begin{align}
		\frac{1}{k}\sum_{i\in \cL \cap \cY}2^m2^{-u(i)\varepsilon}
		&= \frac{1}{k}\sum_{i\in \cL \cap \cL' \cap \cY}
										  2^m2^{-u(i)\varepsilon}
		  +
			\frac{1}{k}\sum_{i\in (\cL \setminus \cL') \cap \cY}
										  2^m2^{-u(i)\varepsilon}\nonumber\\
	  &\stackrel{\text{(iv)}}{\in} 
				\frac{1}{k}\sum_{i\in \cL \cap \cL' \cap \cY}
												2^m2^{-u(i)\varepsilon}
				+ [0, 38\sqrt{\varepsilon}\alpha],
		\label{eqnp:31}
		\\
		\frac{1}{k}\sum_{i\in \cL'\cap \cY'}2^m2^{-u'(i)\varepsilon}
		&= \frac{1}{k}\sum_{i\in \cL \cap \cL'\cap \cY'}
										  2^m2^{-u'(i)\varepsilon}
		  +
			\frac{1}{k}\sum_{i\in (\cL' \setminus \cL) \cap \cY'}
										  2^m2^{-u'(i)\varepsilon}\nonumber \\
	  &\stackrel{\text{(iv)}}{\in} 
				\frac{1}{k}\sum_{i\in \cL \cap \cL'\cap \cY'}
												2^m2^{-u'(i)\varepsilon}
				+ [0, 6\sqrt{\varepsilon}\alpha]
				\label{eqnp:32},
	\end{align}
	where we could apply (iv) because 
	$(\cL \setminus \cL') \cap \cY \subseteq \cL \setminus \cL'
	$ and 
	$(\cL' \setminus \cL) \cap \cY' \subseteq \cL' \setminus 
	\cL$.
	Claim~\ref{claim:GainAndLoss} gives 
	\begin{align}
		\frac{1}{k}\sum_{i\in \cL\cap\cL'\cap \cY'}
					2^m 2^{-u(i)\varepsilon}\nonumber
		=
		\frac{1}{k}\sum_{i\in \cL\cap\cL'\cap \cY'}
					2^m 2^{-u'(i)\varepsilon}
		&+ \Gain_{\cL\cap\cL'\cap \cY'}(u',u) \\
		&- \Loss_{\cL\cap\cL'\cap \cY'}(u',u),\label{eqnp:33}
	\end{align}
	and we get
	\begin{align}
		\Gain_{\cL\cap\cL'\cap \cY'}(u',u)
		&\leq
		\Gain_{\cL\cap\cL'}(u',u)
		\stackrel{\text{(vi)}}{\leq} 4\varepsilon, 
		\label{eqnp:34}\\
		\Loss_{\cL\cap\cL'\cap \cY'}(u',u)
		&\leq
		\Loss_{\cL\cap\cL'}(u',u)
		\stackrel{\text{(vi)}}{\leq} 59\sqrt{\varepsilon}\alpha.
		\label{eqnp:35}
	\end{align}
	Thus we find
	\begin{align*}
		\frac{1}{k}\sum_{i\in \cL \cap \cY}2^m2^{-u(i)\varepsilon}
		&\stackrel{\text{(\ref{eqnp:31})}}{\in} 
		[\frac{1}{k}\sum_{i\in \cL \cap\cL'\cap \cY}
					2^m2^{-u(i)\varepsilon} \pm 38\sqrt{\varepsilon}
					\alpha] \\
		&\stackrel{\text{(\ref{eqnp:30})}}{\subseteq} 
		[\frac{1}{k}\sum_{i\in \cL \cap\cL'\cap \cY'}
					2^m2^{-u(i)\varepsilon} \pm 42\sqrt{\varepsilon}
					\alpha]\\
		&\stackrel{\text{(\ref{eqnp:33}), (\ref{eqnp:34}),
							       (\ref{eqnp:35})}}{\subseteq} 
		[\frac{1}{k}\sum_{i\in \cL \cap\cL'\cap \cY'}
					2^m2^{-u'(i)\varepsilon} \pm 101\sqrt{\varepsilon}
					\alpha]\\
		&\stackrel{\text{(\ref{eqnp:32})}}{\subseteq} 
		[\frac{1}{k}\sum_{i\in \cL'\cap \cY'}
					2^m2^{-u'(i)\varepsilon} \pm 107\sqrt{\varepsilon}
					\alpha]\\
		&\stackrel{\text{(\ref{eqnp:15})}}{\subseteq} 
		[\gYL \pm 112\sqrt{\varepsilon}
					\alpha].
	\end{align*}
	This concludes the proof.	
\end{proof}

\bibliographystyle{alpha}
\bibliography{bibliography}

\end{document}